\documentclass{article}

\usepackage{PRIMEarxiv}

\usepackage[utf8]{inputenc} 
\usepackage[T1]{fontenc}    
\usepackage{hyperref}       
\usepackage{url}            
\usepackage{booktabs}       
\usepackage{amsfonts}       
\usepackage{nicefrac}       
\usepackage{microtype}      
\usepackage{lipsum}
\usepackage{fancyhdr}       
\usepackage{graphicx}       
\usepackage{float} 
\graphicspath{{media/}}     

\usepackage{commath}
\usepackage{bbm}
\usepackage{breqn}
\usepackage{amsmath}
\usepackage{amssymb}
\usepackage{optidef}
\usepackage{mathtools}
\usepackage{amsthm}
\usepackage[ruled,vlined]{algorithm2e}

\usepackage{multirow}
\usepackage{comment}

\newtheorem{theorem}{Theorem}[section]
\newtheorem{lemma}[theorem]{Lemma}
\newtheorem{definition}{Definition}
\newtheorem{proposition}{Proposition}[theorem]

\newcommand{\RN}[1]{%
  \textup{\uppercase\expandafter{\romannumeral#1}}%
}

\newcommand\blfootnote[1]{%
  \begingroup
  \renewcommand\thefootnote{}\footnote{#1}%
  \addtocounter{footnote}{-1}%
  \endgroup
}

\DeclareMathOperator*{\argmax}{arg\,max}

\title{Playing Coopetitive Polymatrix Games \\ with Small Manipulation Cost
}

\author{
  \textbf{Shivakumar Mahesh} \\
  University of Warwick \\
  United Kingdom\\
  \texttt{shivakumar.mahesh@warwick.ac.uk} \\
   \and 
  \textbf{Nicholas	Bishop} \\
  University of Southampton \\
  United Kingdom\\
  \texttt{nb8g13@soton.ac.uk} \\
  \AND 
  \textbf{Le Cong Dinh} \\
  University of Southampton \\
  United Kingdom\\
  \texttt{lcd1u18@soton.ac.uk} \\
  \and
  \textbf{Long	Tran-Thanh} \\
  University of Warwick \\
  United Kingdom\\
  \texttt{long.tran-thanh@warwick.ac.uk} \\
  
}

\newcommand{\BibTeX}{\rm B\kern-.05em{\sc i\kern-.025em b}\kern-.08em\TeX}

\begin{document}
\maketitle

\begin{abstract}
Iterated coopetitive games capture the situation when one must efficiently balance between cooperation and competition with the other agents over time in order to win the game (e.g., to become the player with highest total utility). Achieving this balance is typically very challenging or even impossible when explicit communication is not feasible (e.g., negotiation or bargaining are not allowed).
In this paper we investigate how an agent can achieve this balance to win in iterated coopetitive polymatrix games without explicit communication. 
In particular, we consider a 3-player repeated game setting in which our agent is allowed to (slightly) manipulate the underlying game matrices of the other agents for which she pays a manipulation cost, while the other agents satisfy weak behavioural assumptions. 
Regarding the objective of the manipulating agent, we first consider three different objectives in the coopetition domain, namely: (i) being the one with the highest total payoff (ii) winning (i.e., being the one with the highest total payoff) with the largest possible margin; (iii) winning with the lowest inefficiency ratio. After these three objectives, we showcase the flexibility of the solution concept introduced in this paper by focusing on an objective for the manipulating agent completely outside the coopeition domain namely: (iv) maximizing social welfare.

In our first contribution, we propose a payoff matrix manipulation scheme and sequence of strategies for our agent that provably guarantees that the utility of any opponent would converge to a value we desire. 
We then use this scheme to design winning policies for our agent. 
We also prove that these winning policies can be found in polynomial running time. 
In the third part of the paper, we turn to demonstrate the efficiency of our framework in several concrete coopetitive polymatrix games, and prove that the manipulation costs needed to win are bounded above by small budgets. For instance, in the social distancing game, a polymatrix version of the lemonade stand coopetitive game, we showcase a policy with an infinitesimally small manipulation cost per round, along with a provable guarantee that, using this policy leads our agent to win in the long-run. Note that our findings can be trivially extended to $n$-player game settings as well (with $n > 3$).

\blfootnote{Preprint. Under review.}
\end{abstract}

\section{Introduction}
\noindent
Coopetition, a portmanteau of cooperation and competition, concerns settings in which strategic agents can stand to gain from coordinating with others, but have the underlying motivation of outperforming their peers~\cite{nalebuff1996co}. For example, consider a scenario in which several businesses are competing to become market leaders. Whilst businesses are locked in direct financial competition with each other, it may be beneficial for certain businesses to collaborate and produce a product better than those currently on the market. In such settings, each agent is tasked with a nuanced problem when it comes who they choose to collaborate with and who they choose to compete against. If an agent chooses to play against everyone, then it will likely be surpassed by opposing agents who decide to team up. On the other hand, if an agent chooses to collaborate with everyone, it cannot expect to out-compete anyone. 

Although originally defined within the context of economics, coopetition now plays a central role in multi-agent learning~\cite{lowe2017multi,wen2021cooperative,ryu2021cooperative,bansal2018emergent}, as well as many other areas of multi-agent systems (MAS)~\cite{panait2005cooperative,phan2020learning,yuan2020cooperative, pant2021strategic}. More specifically, coopetition is often formulated as an iterated game between a number of strategic agents, in which the goal of each agent is not to maximise its own cumulative payoff, but achieve the highest cumulative payoff among all participants. In this sense, the goal of an agent is to out-compete its peers and become the ``winner''.
When communication between agents is explicitly feasible, many MAS based approaches can be used to initiate and maintain these cooperation, ranging from negotiation and bargaining theory, to coalitional game theory and coordination~\cite{panait2005cooperative,tan1993multi}.

However, when explicit communication is not feasible, achieving necessary cooperative behaviours becomes a significantly more difficult problem.
Recently, there has been a line of research investigating whether cooperative behaviours can emerge solely from observing and reacting to the past behaviours of opponents~\cite{dinh2021last,bishop2021guide,phan2020learning}.
The key challenge lies not just in identifying the appropriate opponents to cooperate with, but in knowing when to switch sides as well. 
For example, in the lemonade stand game (LSG)~\cite{zinkevich2011lemonade}, three players simultaneously place their stands on one of the twelve positions uniformly distributed on the shore of a circle-shaped island. The payoff of each player is the sum of the distances between their stand and that of their opponents (a more detailed description of its polymatrix game version, called social distancing game, can be found in Section~\ref{sec:social distancing game}). The goal of each agent is then to win the game, that is, to be the one with the highest total payoff over a finite period of time.
In order to achieve this, an agent must pick one of its two opponents and start cooperating (e.g., by placing their stands at the opposite positions of the circle). By doing so, one can easily prove that the average payoff of the two cooperating players will be significantly higher than that of the third. However, this cooperation alone will not provide a guaranteed win, as the cooperating partner may receive higher payoffs, depending on the strategy employed by the other player. Thus, a key step in this game is to know the right time to switch teams and start cooperating with the other player. By doing so, one might be able to become the one with the highest payoff in the long run. 

Note that although the LSG is a rather simplistic setting, it captures the essence of the abovementioned coopetition dilemma in many real-world applications, ranging from technological battles (e.g., the high-definition optical disc format war between Blu-ray and DVD) and R\&D alliances~\cite{baglieri2012asymmetric}, to environmental politics~\cite{carfi2012global}, and multiplayer video gaming~\cite{samvelyan2019starcraft}, where strategically switching sides and its timing are critical.      



In this work, we seek to address the challenges associated with coopetitive multi-agent settings in the following manner:
We relax the original setting by considering the case where one of the players is keen to sacrifice a (small) portion of their received payoff to modify the other agents' payoff values (e.g., the player donates some of their payoffs to the opponents, or makes some costly effort to reduce the others' payoffs). We refer to this type of action as payoff manipulation, the corresponding cost as the manipulation cost, and the player who performs this as the manipulator. 

%
We also consider different types of opponents for the manipulator. These types differ from each other in term of the assumptions made on their behaviour/power, and for each type we present different policies for the manipulator that exploit their respective behavioural assumptions. 
%
In particular, the three successive behavioural assumptions we make on the opponents, from weakest to strongest are:
\begin{itemize}

\item (\textbf{Consistent agent}) An agent who will eventually learn to play a dominant strategy a large fraction of the time, if one exists. By a dominant strategy, we mean a strategy which is the best response on every time step.
\item (\textbf{Persistent agent}) An agent who will learn (and play) the best fixed action in hindsight when there exists a single fixed action which is highly likely to be the best as the time horizon lengthens. For example, persistent agents perform well on sequences where the best fixed action in hindsight eventually becomes fixed and does not change in future time steps. However, a persistent agent makes no guarantees on sequences where the best fixed action in hindsight may vary erratically irrespective of time. 
\item (\textbf{No-regret agent}) An agent whose expected (average) regret can be made arbitrarily close to zero for a sufficiently large time horizon. Note that the performance discrepancy between a no-regret algorithm and the best fixed action in hindsight vanishes (in expectation) as the time horizon lengthens with no assumptions placed on the structure of payoff sequences encountered by the agent. That is, a no-regret agent always performs well regardless of the actions taken by opponents.
\end{itemize} 
One way to interpret the agent types above is by the assumptions they place on the sequences of payoffs they expect to observe. A consistent agent expects there to be a fixed action which is the best response on \emph{every} time step. Meanwhile, a persistent agent expects sequences where, at some point, the best action on average becomes clear and does not change erratically. On the other hand, a no-regret agent makes no assumptions on the sequences of payoffs it may encounter.
%
In Section~\ref{subsec:batch coordination} we will provide a more detailed discussion on the differences of these opponent types (e.g., an agent who employs the Follow the Leader algorithm \cite{Nicolo06} is persistent but not no-regret).
Note that even the strongest assumption we make is mild and widely used in the game theory and online learning communities. 

To highlight the essence of the coopetition dilemma, this work focuses on the iterated polymatrix game setting, which plays an important role in game theory and multiagent systems~\cite{janovskaja1968equilibrium,howson1972equilibria,eaves1973polymatrix,deligkas2017computing,cai2016zero}. 
In particular, the class of polymatrix games is a natural generalisation of the two-player bimatrix game setting and is efficient in capturing pairwise interactions in multiagent systems due to their succinct representation~\cite{cai2016zero,ortiz2015graphical,deligkas2017computing}. It also has wide interest and applicability, especially in the fields of artificial intelligence and machine learning~\cite{miller1991copositive,erdem2012graph,ma2021polymatrix,skoulakis2021evolutionary,liu2021approximation,irfan2014influence}. 

To better emphasize our key technical findings, we concentrate on the 3-player setting for most of this paper. Note that our techniques can be efficiently extended to the general $n$-player setting as discussed in Section~\ref{sec:conclusions}. 
%
%
Against this background, 
our contributions are as follows:

\begin{itemize}
    \item 
    First, we propose a number of winning policies for the manipulator. In particular, we show that there exist a set of policies, which we call dominance solvable, that can guarantee a win for the manipulator (Theorem \ref{thm1}). Moreover, we show that such policies can be found in polynomial time (Theorem \ref{thm3}). Briefly summarised, our approach relies on the fact that all the opponent types we consider will eventually adopt dominant strategies, if they exist. As a result, the manipulator is always able to manipulate payoffs to ensure that all players adopt dominant strategies that are favourable to the manipulator. That is, the manipulator wins when said dominant strategies are played by the opponents.
    
    \item To address situations where employing dominance solvable policies invokes high manipulation costs, we propose a new set of policies, called batch coordination policies, that provably guarantee low manipulation costs (Theorem \ref{thm4}). Similarly to dominance solvable policies, here we show that batch coordination policies can also be computed in polynomial time (Theorem \ref{thm5}). By design, batch coordination policies naturally leverage the idea of switching sides to beat each opponent separately, one at a time, to ensure smaller manipulation costs over the entire time horizon.
   
    \item Aside from the aforementioned objective to win the game, we also show that variations of our proposed policies can be used to achieve many other objectives such as (i) winning by the largest possible margin; (ii) winning with the lowest inefficiency ratio. After these objectives, we showcase the flexibility of the solution concept introduced in this paper by focusing on an objective for the manipulating agent completely different from winning, namely: (iii) maximizing social welfare. In short, it is possible to achieve many other objectives because the policies we outline correspond to linear feasibility problems, which can be transformed into linear programs in order to consider different objectives (see Section~\ref{sec: additional objectives} for more details).
    \item Finally, we show that our method can be applied to many interesting polymatrix games to ensure that the manipulator wins with low manipulation costs (Sections~\ref{sec:social distancing game}-~\ref{section: battle of budies}).
    In particular, we consider the following games: (i) the social distancing game (a polymatrix version of the LSG (Section~\ref{sec:social distancing game}); (ii) the three-player iterated prisoners' dilemma (Section~\ref{section: prisoner's dilemma}); (iii) the electric vs. petrol futures game (Section~\ref{section: electric petrol futures}); and (iv) the battle of the buddies, a polymatrix generalisation of the battle of the sexes (Section~\ref{section: battle of budies}). 
    We show that for all of these games, that when our policies are employed, the manipulation cost incurred is significantly lower than the total payoff the manipulator achieves. Perhaps most surprisingly, in the case of the social distancing game, we show that the manipulator can make an infinitesimally small manipulation and still win (see Section \ref{sec:social distancing game}).
    
\end{itemize}


\subsection{Related Work}
\label{subsec: related work}
Coopetition has been studied extensively in management science~\cite{nalebuff1997co,wang2013advantage,bigliardi2011successful}, supply chains~\cite{chen2018coopetition,gurnani2007impact,trapp2020maritime}, economics~\cite{eriksson2010partnering,loch2006balancing}, and many other related fields (see, e.g.,~\cite{bengtsson2016systematic} for a more thorough literature review). 
In computer science and artificial intelligence, this problem has only been invetigated recently, but it is gaining rapid popularity (especially in multi-agent systems and multi-agent learning)~\cite{lowe2017multi,wen2021cooperative,ryu2021cooperative,bansal2018emergent}.
However, to the best of our knowledge, neither of the abovementioned work have considered manipulation cost (as we do in our paper), and therefore might not be able to find winning policies with small manipulation costs in our setting.
It is worth noting that the closest to our setting is the work from~\cite{bishop2021guide}, which also considers the problem of payoff matrix manipulation so that the unique Nash equilibrium of the new game is a predefined strategy profile. But it does not consider the coopetitive setting as we do here.

Throughout this paper, the policies we propose consist of the manipulator setting up payoff matrices so that players converge to a winning equilibrium. As a result, in implementing the policies we present, the manipulator is tasked with identifying payoff matrices satisfying a prespecified equilibrium condition. Many similar problems have been studied in the literature, under a number of different names, such as equilibrium design, and inverse game theory.

The design of payoff matrices with prespecified Nash equilibrium was first studied by \cite{bohnenblust1950}, who presented a scheme for constructing zero-sum games with a unique prespecified equilbria. The design of games with unique equilibria has since been investigated for broader classes of games, such as bimatrix games \cite{rag, Kreps1974}, and n-player differentiable concave games \cite{rosen65}.
Within the optimisation commuity, a similar problem of certifying the uniqueness of linear programming solutions has been studied extensively. Appa~\cite{appa2002uniqueness} provides a constructive method for verifying the uniqueness of an LP solution which requires solving an addition linear program. Mangasarian~\cite{mangasarian1979uniqueness} describes a number of conditions which guarantee the uniqueness of an LP solution. In fact, it is one of these conditions that we shall leverage to derive our methods for constructing zero-sum games that have unique solutions. More generally, many other works deal with characterising the optimal solution sets of linear programs \cite{Kantor, Tavas}, but do not typically pay any special attention to uniqueness.

Similar problems have also been studied under the name of inverse game theory, in which one is tasked with identifying utility functions that rationalise the action choices made by agents over a time horizon. In the field of optimal control there is an extensive line of work which considers the task of approximating the underlying objective functions of agents in a control setting, given an open-loop Nash equilibrium \cite{molloy, dynnamiccontrol, zhang}. Meanwhile, \cite{succinct} propose algorithms for identifying utility functions for succinct games which induce a prespecified correlated equilibrium. Similarly, \cite{rationalise}, propose methods for identifying correlated equilibria that rationalise the action history of agents participating in a repeated normal form game. Closely related is the field of inverse reinforcement learning, in which the reward function of an agent is estimated given a sample of trajectories \cite{ng}.

Our work differs from those discussed above in a number of respects. First of all, in all our policies, the equilibrium is explicitly chosen by the manipulator, rather than given as part of a problem definition. Secondly, our work is firmly placed in a coopetitive setting where the payoff matrices chosen have a direct effect on the utility of the manipulator. As a result, the manipulator is faced with not only finding payoff matrices that induce a given equilirbium, but also optimising both the payoff matrices and equilibrium chosen so as to minimise its manipulation cost.

\section{Preliminaries}

To begin, we introduce some basic definitions from game theory through which our problem setting will be formally described. We define a finite normal form three-player general-sum game, $\Gamma$, as a tuple $(\mathcal{N},\mathcal{A}, u)$. We denote the set of players by $\mathcal{N} = \{1,2,3\}$. Each player $i \in \mathcal{N}$ must simultaneously select an action from a finite set $\mathcal{A}_i$. For the sake of brevity, we use $n$, $m$ and $l$ to denote the cardinalities of $\mathcal{A}_{1}$, $\mathcal{A}_{2}$ and $\mathcal{A}_{3}$ respectively. We denote by $\mathcal{A} = \mathcal{A}_1 \times \mathcal{A}_2 \times \mathcal{A}_3$ the set of all possible combinations of actions that may be chosen by the players.  

Furthermore, each player is allowed to randomize their choice of action. In other words, player $i$ can select any probability distribution $s \in \Delta(\mathcal{A}_i)$ over her action set. An action is then selected by randomly sampling according to this distribution. We refer to this set of probability distributions as the set of strategies available to the player. We say that a strategy is pure if it corresponds to the deterministic choice a single action, otherwise we say that a strategy is mixed. 
We denote the strategy chosen by player $1$ by the vector $\textbf{x}$, where $\textbf{x}(i)$ indicates the probability that player $1$ selects action $i$. Similarly, we use $\textbf{y}$ and $\textbf{z}$ to denote the strategies chosen by players 2 and 3 respectively.

After strategies have been selected, player $i$ receives a reward given by her utility function $u_i : \mathcal{A} \rightarrow \mathbb{R}$, which we consider to be a random variable under the probability space $(\mathcal{A},\mathcal{F},\mathbb{P})$ where we define the event space $\mathcal{F}$ to be the power set of $\mathcal{A}$ along with the probability measure $\mathbb{P}$ to be the real-valued function $\mathbb{P} : \mathcal{F} \rightarrow [0,1]$ such that for any $a_1 \in \mathcal{A}_1$, $a_2 \in \mathcal{A}_2$ and $a_3 \in \mathcal{A}_3$,  $\mathbb{P}\big(\{(a_1,a_2,a_3)\}\big) = \textbf{x}(a_1) \cdot\textbf{y}(a_2) \cdot\textbf{z}(a_3)$.

In this paper, we restrict our focus to polymatrix games. That is we assume that the utility function for each agent is of the form $u_{i} = \sum_{i \neq j}u_{ij}$, where $u_{ij}: \mathcal{A}_{i} \times \mathcal{A}_{j} \to \mathbb{R}$ describes the payoff player $i$ receives from its interaction with player $j$. Observe that any three-player polymatrix game can be succinctly represented by six payoff matrices $A^{(i, j)}$, which each correspond to a function $u_{ij}$. Additionally, we let $\; \norm{A}_\infty := \max_{k,l}|A(k,l)|$ denote the infinity norm of a given payoff matrix. 


In what follows, we consider a direct extension of the three-player polymatrix setting, in which player 1 takes the role of a manipulator, and is allowed to alter the payoff matrices $A^{(2,1)}$ and $A^{(3,1)}$. In other words, we assume that player $1$ has control over the payoffs other players receive when interacting with her. Thus, in addition to selecting a strategy, player $1$ is also tasked with specifying the payoff matrices $A^{(2, 1)}$ and $A^{(3,1)}$. We refer to the joint submission of a strategy and payoff matrices as player 1's complete strategy. We denote player 1's complete strategy by the tuple $(\textbf{x},(A^{(i,j)})_{(i,j)\in P})$, where $P$ is the index set $\{(2,1),(3,1)\}$.

We use $A_{0}$ to denote the original payoff matrices of the game before they are altered by player $1$. One can interpret $A_{0}$ as a description of the dynamics of interaction between players, before the manipulator has implemented rules and restrictions. In a realistic setting, player 1 should not be able to modify the original game wherever there is interaction between player 2 and 3 alone. We clearly capture this notion in polymatrix games by specifying that player 1 cannot modify the matrices $A_0^{(2,3)}$ and $A_0^{(3,2)}$.

We assume that there is an associated cost for modifying the payoff matrices, which takes the form $\sum_{(i, j) \in P}\|A^{(i, j)} - A_{0}^{(i, j)}\|_{\infty}$. This cost has a natural interpretation when the manipulator uses monetary incentives in attempt to alter the behaviour of fellow players. More specifically, the cost corresponds to the sum of the maximum monetary payments (or fines) each player can receive, and thus represents, in the worst case, how much the manipulator may need to pay (or charge) in order to implement an altered version of the game.  

With this cost in mind, observe that the expected payoff (or utility) of player 1 is given by the expected payoff it receives when participating in the altered polymatrix game, minus the cost it incurs for altering payoff matrices:
\[\textbf{x}^TA^{(1,2)}\textbf{y} + \textbf{x}^TA^{(1,3)}\textbf{z} - \sum_{(i,j) \in P}{\norm{A^{(i,j)}-A_0^{(i,j)}}_{\infty}}.\] 
In contrast, the expected utility of player $2$ is simply given by the expected payoff it receives from participating in the altered game:
\[\textbf{x}^TA^{(2,1)}\textbf{y} + \textbf{y}^TA^{(2,3)}\textbf{z}.\]
Similarly, the expected payoff of player 3 is given by:
\[\textbf{x}^TA^{(3,1)}\textbf{z} + \textbf{y}^TA^{(3,2)}\textbf{z}.\]


Note that since all three players are employing mixed strategies, the payoff observed by each player may not be the same as the expected payoff. For example, if the players sample actions $(a_1,a_2,a_3)$ from the distributions $(\textbf{x}$, $\textbf{y}$ and $\textbf{z})$, then the utility player 1 observes is
\begin{dmath*}
    u_1(\textbf{x},\textbf{y},\textbf{z}) = A^{(1,2)}(a_1,a_2) + A^{(1,3)}(a_1,a_3) - \sum_{(i,j) \in P}{\norm{A^{(i,j)}-A_0^{(i,j)}}_{\infty}}
\end{dmath*}
Similarly, the utility for player 2 is 
\[u_2(\textbf{x},\textbf{y},\textbf{z}) = A^{(2,1)}(a_1,a_2) + A^{(2,3)}(a_2,a_3) \]
and the observed utility for player 3 follows in a analogous manner. When the strategies used are clear from context we will drop them from notation and use ${u}_{1}$, ${u}_{2}$, ${u}_{3}$. We say player 1 has won the game if her utility is higher than the utilities of other players.  


\section{Problem Setting}
\label{sec:problem setting}

In many cases, a manipulator will engage repeatedly with the same system participants. Additionally, aside from the manipulator, players are often unaware of their own, and others, utility functions and must learn them over time. With these concerns in mind, we consider an iterated (or repeated) version of the setting described above. 

More specifically, we consider a setting in which players engage in the aforementioned polymatrix game repeatedly for $T$ time steps. At each time step $t$, each player is required to commit to a strategy, $\textbf{x}_t$, $\textbf{y}_t$ and $\textbf{z}_t$. In addition, player 1, in her role as manipulator, must select the set of payoff matrices $A_{t}^{(i,j)}$, for $(i, j) \in P$, at each time step. 
%
%
We assume that players 2 and 3 have no initial knowledge of $A_{0}$, but receive feedback, at the end of every time step detailing the payoff they received. More precisely, player $2$ receives feedback $u_{2, t} = u_{2}(\mathbf{x}_{t}, \mathbf{y}_{t}, \mathbf{z}_{t})$, at the end of time step $t$. Player 3 receives feedback in a similar fashion. Therefore, when selecting their strategy in round $t+1$, players $2$ and $3$ have access to a history of feedback (and a history of their own strategy choices) up to time step $t$ to inform their decision. In contrast, whilst also receiving feedback at the end of each time step, we assume that player $1$ has full knowledge of $A_{0}$ prior to the start of play. 

We use $H_{t} = (u_{1, t^{\prime}}, \mathbf{x}_{t^{\prime}})^{t}_{t^{\prime}=1}$ to denote the history observed by player $1$ up to time step $t$. We use the notation $\mathcal{H}_{t}$ to denote the set of all observable histories of length $t$. Given a time horizon $T$, we define a policy $\rho = (\rho_{t})_{t=1}^{T} $ as a sequence of, potentially randomized, mappings $\rho_{t}: \mathcal{H}_{t} \to \Delta(\mathcal{A}_{1}) \times \mathbb{R}^{n \times m} \times \mathbb{R}^{n \times l}$ from feedback histories to complete strategies. In other words, a policy $\rho$ is a specification of which complete strategy to choose given the feedback observed so far.


Generalizing from the single-shot setting, we define the utility of each player in the iterated setting as the time average of their respective utilities in each round. That is:
\begin{equation*}
    U_i(\textbf{x}_t,\textbf{y}_t,\textbf{z}_t)_{t=1}^{T} := \frac{1}{T}\sum_{t=1}^{T}{u_i(\textbf{x}_t,\textbf{y}_t,\textbf{z}_t)} 
\end{equation*}
As before, when the sequence of strategies used by each player is clear from context, we write $U_{1}, U_{2}$ and $U_{3}$ for the sake of brevity.\\


\noindent
\textbf{Winning the game}: In what follows, we consider the following version of manipulator objective (and will discuss the other three versions later in Section~\ref{sec: additional objectives}). We say that player $1$ has won the game, if her utility is the highest. We assume that player 1's goal is to win the game. 
On the other hand, we assume that players $2$ and $3$ are ``consistent" agents in the following sense:

\begin{definition}
(\textbf{Consistent Agent}) Assume that for an agent there exists an action $a^*$ that is the unique best response for her for every round of the game. Suppose that within $T$ rounds of the game, the number of rounds the agent plays action $a^*$ is $T^*$.  We say the agent is consistent if \[\mathbb{P}\Big(\lim_{T\rightarrow\infty}{\frac{T^*}{T}} = 1\Big)=1, \] 
\end{definition}
In other words, if an agent has a single action that performs the best in all rounds, and the proportion of time she plays that action converges to 1 almost surely, then we say she is consistent. If a consistent agent does not have an action that always performs the best, we make no assumption on the behaviour of the player.    
Unless stated otherwise, we restrict our focus to players who are consistent, no matter the strategies submitted by the other players. This assumption is much weaker that the standard assumption of rationality in full information mechanism design settings. In the sections that follow, we will develop a number of policies which guarantee player 1 a winning outcome with high probability under this assumption.


\section{Winning Policies}
\label{sec:policies}
In this section, we present a number of policies which guarantee player 1 a winning outcome with high probability. Before describing these policies in detail, we first present a brief conceptual argument showcasing the underlying idea behind all the policies we present. 
Consider the policy where player $1$ plays the same action $i^{*}$ in every round. Assume that player 2 and player 3 each have strictly dominant actions $j^*$ and $k^*$ respectively against, the action $i^*$ of player 1. That is, $u_{2}(\textbf{e}_{i^{*}},\textbf{e}_{j^{*}},\textbf{e}_{k}) > u_{2}(\textbf{e}_{i^{*}},\textbf{e}_{j},\textbf{e}_{k^{}})$ for all $j \neq j^*$ and $k$ and $u_{3}(\textbf{e}_{i^{*}},\textbf{e}_{j},\textbf{e}_{k^{*}}) > u_{3}(\textbf{e}_{i^{*}},\textbf{e}_{j},\textbf{e}_{k})$ for all $j$ and $k \neq k^*$.
Since both players are consistent, the proportion of time each of them plays their strictly dominant action converges to 1 almost surely. 
Therefore, if $u_{1}(\textbf{e}_{i^{*}}, \textbf{e}_{j^{*}}, \textbf{e}_{k^{*}}) \geq \max\{u_{2}(\textbf{e}_{i^{*}},\textbf{e}_{j^{*}},\textbf{e}_{k^{*}}), u_{3}(\textbf{e}_{i^{*}},\textbf{e}_{j^{*}},\textbf{e}_{k^{*}})\}$, then intuitively, player 1 will eventually win if $T$ is large enough. Unfortunately such an action $i^{*}$, which satisfies the above assumptions, may not exist in the original game. However, player 1 can always guarantee the existence of such an action by altering payoff matrices. If player 1 can find a low cost alteration, then she can win the game with high probability. 

We then present another policy of a similar flavor where player $1$ plays the same action $i^{*}$ in every round. We assume that player 2 has a strictly dominant action $j^*$ against the action $i^*$ of player 1, but player 3 only has a strictly dominant action $k^*$ against the action $i^*$ of player 1 and the action $j^*$ of player 2. Therefore if player 3, is an agent who is willing to wait for some action to eventually become her unique best response, then the manipulator can modify the payoff matrices appropriately to ensure that she wins if $T$ is large enough. Therefore in order for such a policy to work successfully, we must make a slightly stronger behavioural assumption on player 3, which leads us to the definition of a persistent agent.

All of the policies we present here combine games constructed to satisfy assumptions similar to those above, with a simple time-dependent deterministic policy. First in Section \ref{sec:designing-dominance-solvable-games} we show how to construct payoff matrices such that actions $j^{*}$ and $k^{*}$ are strictly dominant for players 2 and 3, under the assumption that player 1 uses action $i^{*}$. In Section \ref{sec:dominance-solvable-policies}, we present the class of dominance solvable policies, which consist of stationary policies leveraging the methodologies developed in the previous section. Lastly, we present the class of batch coordination policies, which spend half of the time horizon cooperating with one player, and the other half cooperating with the other.

\subsection{Designing Dominance Solvable Games}
\label{sec:designing-dominance-solvable-games}
Here, we describe several constructions of three player games which will be used extensively in our definitions for different kinds of policies. In particular, we show how to find a matrix $A^{(2,1)}$ (or $A^{(3,1)}$) such that a particular action for player 2 (or 3) is strictly dominant against all actions of player 3 (or 2) and a particular action of the manipulator. We also show how to find a matrix $A^{(3,1)}$ (or $A^{(2,1)}$) such that a particular action for player 3 (or 2) is strictly dominant against a particular action of player 2 (or 3) and a particular action of the player 1. For the sake of brevity we refer to players 1, 2 and 3 by P1, P2 and P3 respectively. 
%
%
Let $\textbf{x}$ be the fixed strategy of the manipulator. To ensure that P2 has a strictly dominant strategy $\textbf{e}_{j^*}$ against $\textbf{x}$ and all actions of P3, For some $v_{2} \in \mathbb{R}^l$, (with $v_{2,k}$ indicating the $k^{\text{th}}$ element of $v_2$) we must choose a matrix $A^{(2,1)}$ that satisfies the system
\begin{equation}
    \label{eq:1}
    \begin{aligned}
        [\textbf{x}^TA^{(2,1)}\textbf{e}_j+\textbf{e}_j^TA^{(2,3)}\textbf{e}_k] =  v_{2,k}\;\; \forall k\in[l] \text{ and } j=j^*\\
        [\textbf{x}^TA^{(2,1)}\textbf{e}_j+\textbf{e}_j^TA^{(2,3)}\textbf{e}_k] <  v_{2,k}\;\; \forall k\in[l] \text{ and } j \neq j^*
    \end{aligned}
\end{equation}


By symmetry, to ensure that P3 has a strictly dominant strategy $\textbf{e}_{k^*}$ against $\textbf{x}$ and all actions of P2, for some $v_3 \in \mathbb{R}^m$ we must choose a matrix $A^{(3,1)}$ that satisfies the system 
\begin{equation}
    \label{eq:2}
    \begin{aligned}
        [\textbf{x}^TA^{(3,1)}\textbf{e}_k+\textbf{e}_j^TA^{(3,2)}\textbf{e}_k] =  v_{3,j}\;\; \forall j\in[m] \text{ and } k=k^*\\
        [\textbf{x}^TA^{(3,1)}\textbf{e}_k+\textbf{e}_j^TA^{(3,2)}\textbf{e}_k] <  v_{3,j}\;\; \forall j\in[m] \text{ and } k \neq k^*
    \end{aligned}
\end{equation}

Now further suppose that P2 plays the fixed strategy $\textbf{y}$. In order to make $\textbf{e}_{k^*}$ the dominant strategy against the strategies $\textbf{x}$ and $\textbf{y}$ of P1 and P2 respectively, for some $v_0 \in \mathbb{R}$ we must choose a matrix $A^{(3,1)}$ 
that satisfies the system
\begin{equation}
    \label{eq:3}
    \begin{aligned}
    [\textbf{x}^TA^{(3,1)}\textbf{e}_k+\textbf{y}^TA^{(3,2)}\textbf{e}_k] =  v_{0}\;\; k=k^*\\
    [\textbf{x}^TA^{(3,1)}\textbf{e}_k+\textbf{y}^TA^{(3,2)}\textbf{e}_k] <  v_{0}\;\; k \neq k^*
    \end{aligned}
\end{equation}

The following lemma shows that strategy profiles satisfying systems \eqref{eq:1} and \eqref{eq:2} always exist.

\begin{proposition}
\label{prop1}
Fix $i^* \in [n], \; j^* \in [m] \text{ and } k^* \in [l]$ with $\textup{\textbf{x}} = \textbf{e}_{i^*}$. Matrices $A^{(2,1)}$ and $A^{(3,1)}$ that satisfy the systems \eqref{eq:1} and \eqref{eq:2} exist.
\end{proposition}

In addition, this result clearly extends to system \eqref{eq:3}, as any matrix satisfying system \eqref{eq:2} satisfies system \eqref{eq:3}.
\begin{proposition}
\label{cor1}
Fix $i^* \in [n], \; j^* \in [m] \text{ and } k^* \in [l]$ with $\textup{\textbf{x}} = \textbf{e}_{i^*}$ and $\textup{\textbf{y}} = \textbf{e}_{j^*}$. Matrices $A^{(2,1)}$ and $A^{(3,1)}$ that satisfy the systems systems \eqref{eq:1} and \eqref{eq:3} exist.
\end{proposition}

In what follows, we develop policies based on payoff matrices which satisfy systems \eqref{eq:1}-\eqref{eq:3}.
\subsection{Dominance Solvable Policies}
\label{sec:dominance-solvable-policies}
In this section, we introduce the class of dominance solvable policies. In short, dominance solvable policies consist of player 1 playing a constant complete strategy which satisfies a number of the linear systems outlined above. We first introduce type-\RN{1} dominance solvable policies.

\begin{definition}
(\textbf{Dominance Solvable Type-\RN{1} Policy}) Let $\left(A^{(2, 1)}, A^{(3, 1)}\right)$ satisfy systems \eqref{eq:1} and \eqref{eq:2} for some $i^{*} \in [n]$, $j^{*} \in [m]$, $k^{*} \in [l]$. Then, the policy  $\rho_t(H_t) = \left(e_{i^{*}}, A^{(2, 1)}, A^{(3, 1)}\right)$ for $t \in \mathbb{N}$ 
is a dominance solvable type-\RN{1} policy.
\end{definition}

That is, a dominance solvable type-\RN{1} policy is one in which player $1$ plays a constant complete strategy which satisfies systems \eqref{eq:1} and \eqref{eq:2}. Similarly, we define dominance solvable type-\RN{2} policies as those in which player $1$ plays a constant complete strategy which satisfies systems \eqref{eq:1} and \eqref{eq:3}.

\begin{definition}
(\textbf{Dominance Solvable Type-\RN{2} Policy}) Let $\left(A^{(2, 1)}, A^{(3, 1)}\right)$ satisfy systems \eqref{eq:1} and \eqref{eq:3} for some $i^{*} \in [n]$, $j^{*} \in [m]$, $k^{*} \in [l]$. Then, the policy $\rho_t(H_t) = \left(e_{i^{*}}, A^{(2, 1)}, A^{(3, 1)}\right)$ for $t \in \mathbb{N}$ 
is a dominance solvable type-\RN{2} policy.
\end{definition}


If one uses iterated elimination of strictly dominated strategies and there is only one strategy left for each player, the game is called dominance solvable ~\cite{myerson91gametheorybook}. We name the policies described above dominance solvable since the underlying single-shot game that results from these policies is almost dominance solvable. In the game that results from these policies, if we eliminate all the actions of player 1 except $i^*$ and then implement iterated elimination of strictly dominated strategies, there will be only one strategy left for each player. 

We say that a dominance solvable policy is winning if player $1$ wins the corresponding single-shot game when $(\textbf{e}_{i^{*}}, \textbf{e}_{j^{*}}, \textbf{e}_{k^{*}})$ is played:
\begin{equation}
    \label{eq:4}
    \begin{aligned}
       u_1(\textbf{e}_{i^*}, \textbf{e}_{j^*}, \textbf{e}_{k^*}) \geq u_2(\textbf{e}_{i^*}, \textbf{e}_{j^*}, \textbf{e}_{k^*}),\\
       u_1(\textbf{e}_{i^*}, \textbf{e}_{j^*}, \textbf{e}_{k^*}) \geq u_3(\textbf{e}_{i^*}, \textbf{e}_{j^*}, \textbf{e}_{k^*}) 
    \end{aligned}
\end{equation}



Winning dominance solvable type-\RN{1} policies are highly attractive as they allow the manipulator to win in the long-run against consistent agents. This claim is formalized in the following theorem.

\begin{theorem}
\label{thm1}
If the manipulator uses a winning dominance solvable type-\RN{1} policy against consistent agents in an infinitely repeated game then,  \begin{align*}
& \mathbb{P}\Big(U_1(\textbf{x}_t,\textbf{y}_t,\textbf{z}_t)_{t=1}^{\infty} \geq U_2(\textbf{x}_t,\textbf{y}_t,\textbf{z}_t)_{t=1}^{\infty} \text{ and } \\
& U_1(\textbf{x}_t,\textbf{y}_t,\textbf{z}_t)_{t=1}^{\infty} \geq U_3(\textbf{x}_t,\textbf{y}_t,\textbf{z}_t)_{t=1}^{\infty} \Big) = 1
\end{align*}
\end{theorem}

At times, for the sake of brevity, we use $U_{1}^{\infty}$, $U_{2}^{\infty}$ and $U_{3}^{\infty}$ to denote the long-run utilities $U_1(\textbf{x}_t,\textbf{y}_t,\textbf{z}_t)_{t=1}^{\infty}$, $U_2(\textbf{x}_t,\textbf{y}_t,\textbf{z}_t)_{t=1}^{\infty}$ and $U_3(\textbf{x}_t,\textbf{y}_t,\textbf{z}_t)_{t=1}^{\infty}$ respectively.

For the analogous guarantee on type-\RN{2} policies we assume a slightly stronger behavioural assumption than being 'consistent' on one of the agents. We assume that player 3 is persistent, i.e. if there is some finite-time cutoff point after which there exists an action that always remains the unique best-response in hindsight then she will play that action a large fraction of time.

\begin{definition}
(\textbf{Persistent Agent}) Assume that for an agent there exists an action $k^*$ for which there exists a finite-time cutoff point $T_h$ after which for every single time-step after $T_h$, $k^*$ is the unique best-response in hindsight for this agent. For player 3 this is the condition that,  \[\exists k^{*} \in [l]: \mathbb{P}\Big(\exists T_h \in \mathbb{N}: \forall T \geq T_h : k^* = \argmax_{k \in [l]}{U_3(\textbf{x}_{t},\textbf{y}_{t},\textbf{e}_k)_{t=1}^{T}}\Big)=1\]  Let $T^*$ denote the number of rounds the agent plays action $k^*$ until $T$. We say that the agent is persistent if \[\mathbb{P}\Big(\lim_{T\rightarrow\infty}{\frac{T^*}{T}} = 1\Big)=1\].
\end{definition}
Note that every persistent agent is consistent. We prove this in Proposition \ref{prop2}. The guarantee of winning when using type-\RN{2} policies is exactly the same as type-\RN{1} policies except that we assume one of the players is persistent.

\begin{theorem}
\label{thm2}
If the manipulator uses a winning dominance solvable type-\RN{2} policy against a consistent player 2 and a persistent player 3 in an infinitely repeated game then, \begin{align*}
& \mathbb{P}\Big(U_1(\textbf{x}_t,\textbf{y}_t,\textbf{z}_t)_{t=1}^{\infty} \geq U_2(\textbf{x}_t,\textbf{y}_t,\textbf{z}_t)_{t=1}^{\infty} \text{ and } \\
& U_1(\textbf{x}_t,\textbf{y}_t,\textbf{z}_t)_{t=1}^{\infty} \geq U_3(\textbf{x}_t,\textbf{y}_t,\textbf{z}_t)_{t=1}^{\infty} \Big) = 1
\end{align*}
\end{theorem}

Observe that any constant complete strategy is dominance solvable as long as it satisfies the linear systems \eqref{eq:1} and \eqref{eq:2} (or \eqref{eq:3}) for a given triple of actions $(i^{*}, j^{*}, k^{*})$. Furthermore, a dominance solvable policy is winning if and only if it satisfies the pair of linear inequalities in system \eqref{eq:4}. As a result, winning dominance solvable policies, if they exist, can be found in polynomial time by solving a sequence of linear feasibility problems, where each linear feasibility problem corresponds to a different triple of actions.

\begin{theorem}
\label{thm3}
If winning dominance solvable policies exist, then there exists an algorithm that can find such policies with running time that is polynomial in the number of actions of the players.
\end{theorem}




If player 1 uses a type-\RN{1} policy, she can make a very weak behavioural assumption on the other players to guarantee winning in the long run. On the other hand, type-\RN{2} policies are guaranteed to be at least as cost-effective as type-\RN{1} policies, as all type-\RN{1} policies are also type-\RN{2} policies.


\subsection{Batch Coordination Policies}
\label{subsec:batch coordination}

Note that, even if a winning dominance solvable policy exists, it may be very costly to alter the payoff matrix. However, the manipulator may be able to beat one player through very cheap alterations whilst losing to the other, and vice versa. In this case it makes sense for player 1 to divide the time horizon, spending half the horizon winning over one player, and spending the other half winning over the other, using cheap alterations to the original payoff matrices in the process. This is the central idea behind batch coordination policies. The following definition makes this idea rigorous.


\begin{definition}
\label{def:batch}
(\textbf{Winning Batch Coordination policy}) Suppose the matrices $\hat{A}^{(2,1)}$ and $\hat{A}^{(3,1)}$ satisfy systems \eqref{eq:1} and \eqref{eq:2} for some $i_2 \in [n]$, $j_2 \in [m]$, $k_2 \in [l]$ and that the matrices $\tilde{A}^{(2,1)}$ and $\tilde{A}^{(3,1)}$ satisfy systems \eqref{eq:1} and \eqref{eq:2} for some $i_3 \in [n]$, $j_3 \in [m]$, $k_3 \in [l]$ such that for $i \neq 1$
\begin{align*}
&\mathbb{E}[u_1(\textbf{e}_{i_2},\textbf{e}_{j_2},\textbf{e}_{k_2})]+ \mathbb{E}[u_1(\textbf{e}_{i_3},\textbf{e}_{j_3},\textbf{e}_{k_3})] >\mathbb{E}[u_i(\textbf{e}_{i_3},\textbf{e}_{j_3},\textbf{e}_{k_3})] + \mathbb{E}[u_i(\textbf{e}_{i_3},\textbf{e}_{j_3},\textbf{e}_{k_3})]
\end{align*}
then we define the winning batch coordination policy as follows:
\begin{equation*}
    \rho_{t} = 
    \begin{cases}
        (\textbf{e}_{i_1}, \hat{A}^{(2,1)}, \hat{A}^{(3,1)}) & \textrm{if } \: 1 \leq t \leq T/2\\
        (\textbf{e}_{i_2}, \tilde{A}^{(2,1)}, \tilde{A}^{(3,1)}) & \textrm{if } \: T/2 < t \leq T
    \end{cases}
\end{equation*}
\end{definition}

Winning batch coordination policies can be interpreted as following different dominance solvable policies for each half of the game. Therefore, winning batch coordination policies are more general than winning dominance solvable policies. Note that the dominance solvable polices played in each half of the time horizon may not be winning by themselves. However, when combined, these sub-policies must form a winning policy for the overall batch coordination policy to be winning.

Before we present the guarantee for player 1 when using winning batch coordination policies, we make a slightly stronger behavioural assumption on both players than the assumption of being 'persistent'. We now assume that both players aim to maximize their expected utility. We use the well-established notion of regret as a metric for measuring the performance of players 2 and 3 with respect to the payoffs they accumulate over time.
\begin{definition}
(\textbf{No-regret agent}) The regret of any sequence of strategies $(\textbf{y}_1,...,\textbf{y}_T)$ chosen by player $i \in\{2,3\}$ with
respect to a fixed strategy $\textbf{y}$ is given by 
\[\mathcal{R}_{T,\textbf{y}} = \sum_{t=1}^{T}{\textbf{x}_t^TA^{(i,1)}_t\textbf{y}_t + \textbf{y}_t^TA^{(i,j)}_0\textbf{z}_t} - \sum_{t=1}^{T}{\textbf{x}_t^TA^{(i,1)}_t\textbf{y} + \textbf{y}^TA^{(i,j)}_0\textbf{z}_t}\]
where $j \in\{2,3\}$ but $i \neq j$.
That is, the regret is the difference between the payoff accumulated by the sequence $(\textbf{y}_1,...,\textbf{y}_T)$ and
the payoff accumulated by the sequence where a given fixed mixed strategy $\textbf{y}$ is chosen at each time
step. 
A player is 'no-regret' if her regret with respect to the sequence of strategies chosen by the other two players is sublinear in T:
\[\lim_{T \rightarrow\infty}{\max_{\textbf{y} \in \Delta_m}{\frac{\mathcal{R}_{T,\textbf{y}}}{T}}} = 0.\]
\end{definition}

We present the following proposition that states all persistent players are consistent, and that all no-regret players are persistent.
\begin{proposition}
\label{prop2}
All persistent players are consistent. Further, all no-regret players are persistent.
\end{proposition}
That is, each assumption on the behaviour of the agents is successively stronger. To emphasize this, we prove that there exists a type of player who is persistent but not no-regret.
\begin{proposition}
\label{prop3}
The agent that uses the Follow the Leader algorithm is persistent but not no-regret.
\end{proposition}

If the manipulator uses a winning batch coordination policy against no-regret players, then the probability that there exists some finite number of rounds in which she wins is 1. This result is formalized as follows:

\begin{theorem}
\label{thm4}
If the manipulator uses a winning batch coordination policy against no-regret players then
$\mathbb{P}\Big(U_1(\textbf{x}_t,\textbf{y}_t,\textbf{z}_t)_{t=1}^{T} \geq U_2(\textbf{x}_t,\textbf{y}_t,\textbf{z}_t)_{t=1}^{T} \text{ and }
U_1(\textbf{x}_t,\textbf{y}_t,\textbf{z}_t)_{t=1}^{T} \geq U_3(\textbf{x}_t,\textbf{y}_t,\textbf{z}_t)_{t=1}^{T} \text{ eventually}\Big) = 1
$.
\end{theorem}

It is worth noting that this guarantee on the convergence of utilities is stronger than the one given for winning dominance solvable policies in Theorem \ref{thm1}. This is because of the strict inequality on the utilities of the players in Definition \ref{def:batch}. By presenting this guarantee instead of a guarantee on the infinite horizon utilities, we are ensured that the guarantee of winning in a finite number of rounds with probability 1 implies that the manipulator can use one set of game matrices for half the rounds and another set for the second half.

Similarly to winning dominance solvable policies, winning batch coordination policies can be found in polynomial time by solving a number of linear feasibility problems.
\begin{theorem}
\label{thm5}
If winning batch coordination policies exist, then there exists an algorithm that can find such policies with running time that is polynomial in the number of actions of the players.
\end{theorem}

For the remainder of the paper we use the weakest assumption, that players are consistent but not necessarily persistent.

\section{Additional Objectives}
\label{sec: additional objectives}

The manipulator may have additional goals and objectives aside from simply winning the game. For example, the manipulator may want to win by a large margin, or win by making the smallest alterations to the payoff matrices possible, or even have a goal completely different to winning, such as maximizing the egalitarian social welfare. For each of the policy classes from Section \ref{sec:policies}, the manipulator can solve a sequence of linear feasibility problems in order to find a winning policy if one exists. As long as the linear constraints of one of these problems are satisfied, the manipulator is guaranteed to win (i.e. she has found a winning policy). Therefore, the manipulator can specify any additional objectives she may have as a linear function to optimize with respect to the linear constraints imposed by the policy class. In other words, the manipulator may choose a linear objective function which captures her additional goals, and solve a sequence of linear programs (LPs), instead of a sequence of linear feasibility problems.  


For example let $d_2$ and $d_3$ be the cost of altering matrices $A_0^{(2,1)}$ and $A_0^{(3,1)}$ respectively. If we consider a minimization problem with objective $d_{2} + d_{3}$, then this amounts to finding a winning policy which makes the least cost modification possible. Similarly, let $v_2$ be the payoff for player 2 and $v_3$ be the payoff for player 3 in the strategy profile of consideration. Setting $v_{2}$ as a maximization objective amounts to winning whilst ensuring player 2 does as well as possible. We could also act adversarially against player 2, by instead minimizing $v_{2}$. Meanwhile, setting $v_{2} + v_{3}$ as a maximization objective corresponds to winning whilst maximizing the utilitarian welfare of the other players. 

In what follows, we investigate additional objectives and goals of wider interest. In Section \ref{sec:winning-by-the-largest-margin} we investigate how player 1 can maximize her margin of victory, in Section \ref{winning-with-the-lowest-inefficiency-ratio} we investigate how the manipulator may win in the most cost efficient way possible. Meanwhile, in Section \ref{sec:maximizing-the-egalitarian-social-welfare} we take a brief step outside of the coopetitive setting and investigate how the manipulator may maximize the egalitarian social welfare.




\subsection{Winning by the Largest Margin}
\label{sec:winning-by-the-largest-margin}
In strictly competitive settings, it is often desirable for players to win, whilst ensuring that their long run utility is much higher than the other players. This motivates the following definition:

\begin{definition}
The \textbf{margin} of a policy $\rho_t(H_t) = (\textup{\textbf{x}}_t,(A^{(i,j)}_t)_{(i,j)\in P})$ for $t \in \mathbb{N}$ when playing against player 2 and player 3's no-regret sequence of strategies $(\textbf{y}_t)_{t=1}^{\infty}$ and $(\textbf{z}_t)_{t=1}^{\infty}$ is defined to be \[\min\Big\{\mathbb{E}\left[U^{\infty}_1 - U^{\infty}_2\right],\mathbb{E}\left[U^{\infty}_1 - U^{\infty}_3\right]\Big\}\]
\end{definition}

That is, the margin is the minimum difference between the long run expected utility of player 1 and another player. Any winning dominance solvable policy will have a margin of at least zero. Additionally, for any of the policy classes discussed above, if a winning policy exists, then a winning policy with the largest margin can be found efficiently via the addition of a linear objective and a small number of linear constraints and variables.
\begin{theorem}
\label{thm6}
If winning dominance solvable policies exist, then there exists an algorithm that can find the largest margin dominance solvable policy, with running time that is polynomial in the number of actions of the players.
\end{theorem}

\subsection{Winning with the Lowest Inefficiency Ratio}
\label{winning-with-the-lowest-inefficiency-ratio}
In many scenarios, it is only sensible to make changes to payoff matrices if one would see a large relative improvement compared to the cost of alteration. We characterize the notion of relative improvement using the following definition.     

\begin{definition}
The \textbf{Inefficiency Ratio} of a policy $\rho_t(H_t) = (\textup{\textbf{x}}_t,(A^{(i,j)}_t)_{(i,j)\in P})$ for $t \in \mathbb{N}$ when playing against player 2 and player 3's no-regret sequence of strategies $(\textbf{y}_t)_{t=1}^{\infty}$ and $(\textbf{z}_t)_{t=1}^{\infty}$ is  
\begin{equation*}
\frac{\lim_{T \rightarrow\infty}\frac{1}{T}\sum_{t=1}^{T}\sum_{(i,j) \in P}{\|A^{(i,j)}_t-A_0^{(i,j)}\|_{\infty}}}{
\mathbb{E}\Big[\lim_{T \rightarrow\infty}\frac{1}{T}{\sum_{t=1}^{T}{\left(\textup{\textbf{x}}_t^TA^{(1,2)}_t\textup{\textbf{y}}_t + \textup{\textbf{x}}_t^TA^{(1,3)}_t\textup{\textbf{z}}_t\right)}}\Big] - 
K}
\end{equation*}
where $K = \min_{i,j,k}{\big(A^{(1,2)}(i,j) + A^{(1,3)}(j,k)\big)}$ is the minimum revenue for player 1. 
\end{definition}
In other words, the inefficiency ratio is the ratio between the cost for modifying the payoff matrices and the expected increase in long run payoffs from the worst case payoff. Note that this fraction must converge for the definition to be meaningful. In a similar fashion to maximizing the margin of victory, policies which minimize the inefficiency ratio can be found in polynomial time.
\begin{theorem}
\label{thm7}
If winning dominance solvable policies exist, then there exists an algorithm that can find the winning dominance solvable policy with the lowest inefficiency ratio, with running time that is polynomial in the number of actions of the players.
\end{theorem}

\subsection{Maximizing the Egalitarian Social Welfare}
\label{sec:maximizing-the-egalitarian-social-welfare}
We now consider an altruistic goal for the manipulator that is completely different from the original goal of winning. Here, we develop a policy for the manipulator that ensures the utility of all players are as large as possible. To further this notion, we define the quantity we call the egalitarian social welfare, which we aim to maximize.

\begin{definition}
The \textbf{Egalitarian Social Welfare} of a strategy profile $(\textbf{x},\textbf{y},\textbf{z})$ is defined to be
\begin{equation*}
    \mathcal{S}(\textbf{x},\textbf{y},\textbf{z}) := \min\Big\{U_1(\textbf{x},\textbf{y},\textbf{z}),U_2(\textbf{x},\textbf{y},\textbf{z}),U_3(\textbf{x},\textbf{y},\textbf{z})\Big\}
\end{equation*}
\end{definition}

We can find the dominance solvable policy that maximizes egalitarian social welfare in polynomial running time. Note that such a policy will always exist.  
\begin{theorem}
\label{thm8}
There exists an algorithm that can find the dominance solvable policy that maximizes egalitarian social welfare with running time that is polynomial in the number of actions of the players.
\end{theorem}

\section{Social Distancing Game}
\label{sec:social distancing game}

Next, we consider a more practical application of the theory above. In particular, we present a number of examples of the theory we have developed. Each example highlights a different aspect of the theory. 
In the first application,  we consider the social distancing game which is a small variation of the lemonade stand game introduced by ~\cite{zinkevich2011lemonade} (due to space limitations, we defer the analysis of other games to the appendix).

\begin{quotation}
It is summer on a remote island, and you need to survive. You
decide to set up camp on the beach (which you may shift anywhere around the island), as do two others. There are twelve places to set up around the island like the numbers on a clock. The game is repeated. Every night, everyone moves under cover of darkness (simultaneously). There is no cost to move. The pandemic is eternal, so the game is infinitely repeated. The utility
of the repeated game is the time-averaged utility of the single-shot games. The only person that survives is the one with the highest total utility at the end of the game.
\end{quotation}

\begin{figure}[h]
    \centering
    \includegraphics[width=25mm,scale=0.5]{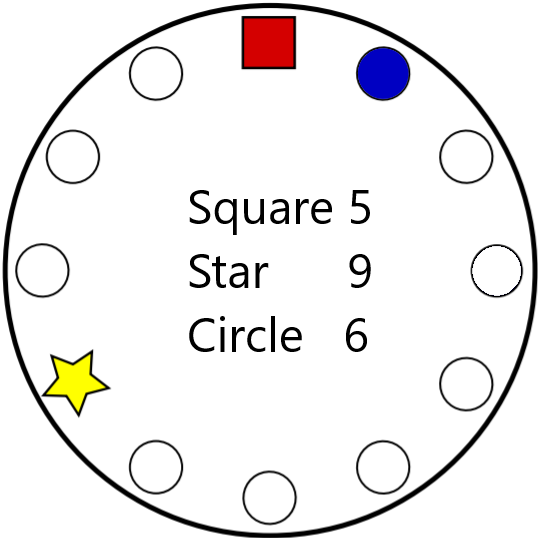}
    \caption{Example Social Distancing Game}
    \label{fig:1sdg}
\end{figure}

\begin{figure}[h]
    \centering
    \includegraphics[width=70mm,scale=0.5]{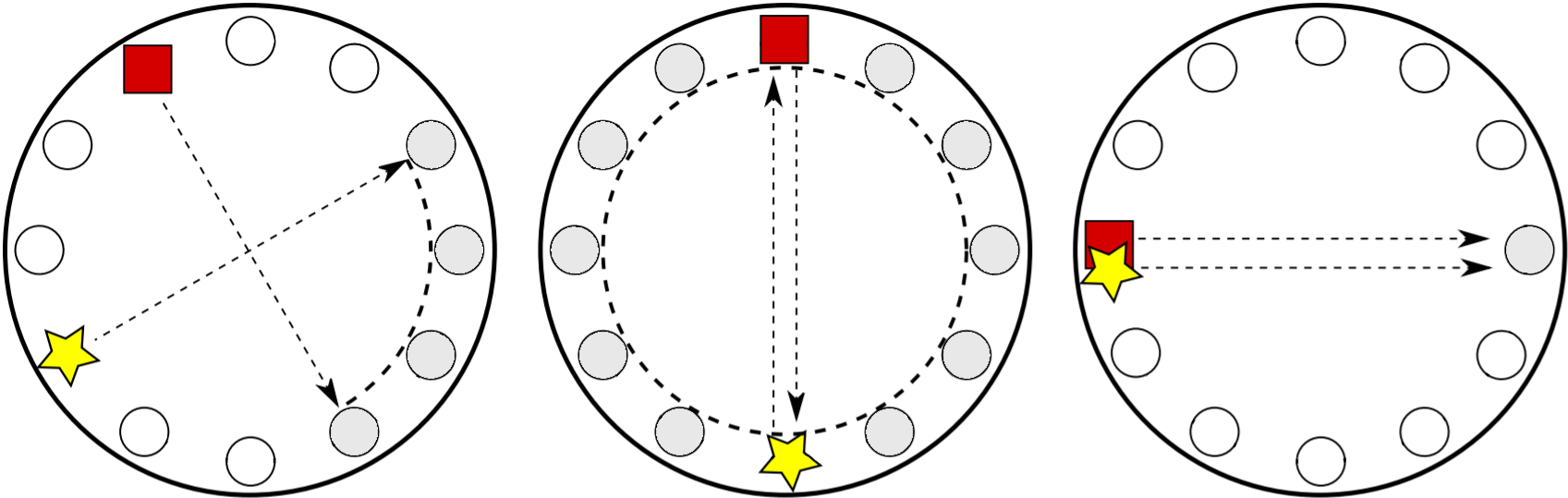}
    \caption{Best-responses for different opponent configurations: The dashed and shaded segment indicates the third player's best-response actions, and arrows point to the action opposite each opponent. (Figures reworked from~\cite{sykulski10})}
    \label{fig:2sdg}
\end{figure}

The utility of a player in a single round of the social distancing game is the sum of its distances from the other two players. The distance between two players is the length of the shortest path between them along the circumference of the clock. More formally, the distance between two positions is defined as follows:
 \[ d(i,j) = \begin{cases} 
          |i-j| &  |i-j| \leq 6\\
          12 - |i-j| & otherwise \\
          
       \end{cases}
\]

For example, if Alice sets up at the 3 o'clock location, Bob sets up at 10 o’clock, and Candy sets up at 6 o’clock, then the utility of Alice is $d(3,10) + d(3,6) = 5 + 3 = 8$, the utility of Bob is $d(10,3) + d(10,6) = 5 + 4 = 9$, and the utility of Candy is $d(6,3) + d(6,10) = 3+ 4 = 7$. If all the camps are set up in the same spot, everyone gets 0. If exactly two camps are located at the same spot, the two collocated camps get the distance to the non-collated camp as their utility and the non-collated camp gets twice the same distance as her utility. In contrast to the lemonade stand game, the social distancing game is not constant-sum. However, it is a polymatrix game, and thus the techniques developed above can be applied.

In what follows, we consider a three-player, infinitely repeated version of the social distancing game . Each player $i$ has 12 actions, $\mathcal{A}_{i} = \{1, \dots, 12\}$, each corresponding to a number on the clock. The payoff matrix for each pair of players $(i, j)$ is derived directly from the distance function $d$. That is, $A_{0}^{(i, j)}(k, l) = d(k, l) $ for all $k, l \in \mathcal{A}_{i}$.



\subsection{Winning Strategy for a Manipulator}
We now present a winning dominance solvable type-\RN{1} policy for the social distancing game. By definition, for any pair of positions $(k, l)$ on the clock, $d(k, l) \leq 6$. This implies that the maximum utility achievable by any player is $12$. In addition, a player $i$ only achieves their maximum payoff when both remaining players place themselves directly opposite of player $i$. Thus, there are only 12 combinations of pure strategies which maximize the utility of player 1, each corresponding to a single number on the clock. In particular, we choose to work with one such strategy profile, $(\mathbf{e}_{12}, \mathbf{e}_{6}, \mathbf{e}_{6})$. 
%
%
Consider the following dominance solvable policy. Set
\begin{equation*}
    \hat{A}(k, l) = \begin{cases} 
                   d(k,l) - \epsilon &   \text{ if } d(k,l) < 6\\
                   d(k,l) + \epsilon &   \text{ if } d(k,l) = 6 \\
                  \end{cases}
\end{equation*}
and let player $1$ adopt the policy $\rho_{t} = \left(\mathbf{e}_{12}, \hat{A}, \hat{A} \right)$. First, observe that, under policy $\rho$, $\mathbf{e}_{6}$ is a dominant strategy for player 2 against the fixed strategy $\mathbf{e}_{12}$ of player $1$. Additionally, by symmetry, $\mathbf{e}_{6}$ is also a dominant strategy for player 3 against the fixed strategy of player 1. Moreover, note that player 1's utility under the strategy profile $(\mathbf{e}_{12}, \mathbf{e}_{6}, \mathbf{e}_{6})$ is $12 - 2\epsilon$. Meanwhile, the utilities of both players 2 and 3 is $6 + \epsilon$. Thus, by Theorem \ref{thm1}, for sufficiently small $\epsilon$ we have 
${\mathbb{P}\left(U_{1}^{\infty} \geq U_{2}^{\infty} \text{ and } U_{1}^{\infty} \geq U_{3}^{\infty}\right)} = 1$. 
Note that such a result implies that player $1$ can guarantee her maximum payoff in long run expectation by only making an infinitesimal change to the payoff matrices!
\subsection{Maximizing Egalitarian Social Welfare}

We now present a socially good solution a manipulator can guide the players to converge to by using a winning dominance solvable policy. In the standard version of the game without a manipulator, one of the "socially optimal" strategy profiles is $(\textbf{e}_{12},\textbf{e}_4,\textbf{e}_8)$, since in this profile, all the players are spread out evenly around the clock. It is possible for a manipulator to guide the players to an approximately optimal solution, in the sense that she can enable convergence to the strategy profile $(\textbf{e}_{12},\textbf{e}_5,\textbf{e}_7)$.
Consider the following dominance solvable policy. Set
\begin{equation*}
    \hat{A}(k, l) = \begin{cases} 
                   d(k,l)                 &   \text{ if } k \neq 12\\
                   d(k,l) - 1 - 2\epsilon &   \text{ if } k = 12 \text{ and } l \neq 5\\
                   d(k,l) + 1 - \epsilon  &   \text{ if } k = 12 \text{ and } l = 5 \\
                  \end{cases}
\end{equation*}
and
\begin{equation*}
    \tilde{A}(k, l) = \begin{cases} 
                   d(k,l)                 &   \text{ if } k \neq 12\\
                   d(k,l) - 1 + \epsilon  &   \text{ if } k = 12 \text{ and } l \neq 7\\
                   d(k,l) + 1 - \epsilon  &   \text{ if } k = 12 \text{ and } l = 7 \\
                  \end{cases}
\end{equation*}

and let player $1$ adopt the policy $\rho_{t} = \left(\mathbf{e}_{12}, \hat{A}, \tilde{A} \right)$. First, observe that, under policy $\rho$, $\mathbf{e}_{5}$ is a dominant strategy for player 2 against the fixed strategy $\mathbf{e}_{12}$ of player $1$. Additionally, $\mathbf{e}_{7}$ is a dominant strategy for player 3 against the fixed strategy of player 1. Moreover, note that player 1's utility under the strategy profile $(\mathbf{e}_{12}, \mathbf{e}_{5}, \mathbf{e}_{6})$ is $10 - (2+\epsilon) = 8-\epsilon$. Meanwhile, the utilities of both players 2 and 3 is also $8 - \epsilon$. Thus, by Theorem \ref{thm1}, for sufficiently small $\epsilon>0$ we have 
${\mathbb{P}\left(U_{1}^{\infty} \geq U_{2}^{\infty} \text{ and } U_{1}^{\infty} \geq U_{3}^{\infty}\right)} = 1$.
Note that such a result implies that player $1$ can guarantee that the game converges to an approximately socially optimal solution whilst ensuring that she still wins the game!

\section{Three-Player Iterated Prisoner's Dilemma}
\label{section: prisoner's dilemma}

The next example we consider is a three-player version of the iterated prisoner's dilemma. As in the two-player version, each player must choose from a set of two actions $\mathcal{A} = \{\text{C, D}\}$ which stand for cooperate and defect respectively. The payoff matrices for each player are defined as follows:
\[A^{(i,j)}_0 = 
\begin{bmatrix}
    3 & 0\\
    5 & 1\\
\end{bmatrix} \text{ if } i < j \;\; \text{and} \;\;
A^{(i,j)}_0 = \begin{bmatrix}
    3 & 5\\
    0 & 1\\
\end{bmatrix} \text{ if } i > j
\]
\subsection{Winning Strategy for a Manipulator}
Note that defection is a strictly dominant strategy for each player. Moreover, the payoff awarded to each player is the same when everyone defects. As a result, by Theorem \ref{thm1}, player 1 can win the game with high probability by repeatedly defecting, and never altering payoff matrices. Note that this policy is zero cost in the sense that the manipulator never needs to alter any payoff matrices. However, the margin is also zero. We now illustrate how alterations to the payoff matrices can result in a winning policy for the manipulator, which has positive margin, and encourages cooperation between players.  
In particular, we outline a policy which the manipulator may use to converge to the strategy profile $(D,C,C)$. For $0 \leq \epsilon \leq 7/6$ we set
\begin{equation*}
     \hat{A} = 
        \begin{bmatrix}
            3 & 5\\
            3/2 + \epsilon & -1/2\\
        \end{bmatrix}.
\end{equation*}


  
Let player $1$ adopt the policy $\rho_{t} = \left(\textbf{e}_2, \hat{A}, \hat{A}\right)$. Note that the mixed strategy of any player is characterized by the probability  that they cooperate. If player $3$ cooperates with probability $\lambda$ then the expected utility player $2$ receives from cooperating is $3/2 + \epsilon + 3\lambda$. Meanwhile, the expected utility player $2$ receives by defecting is $1/2 + 4\lambda$. Since, $\lambda \in [0, 1]$, this implies that cooperation is a strictly dominant strategy for player $2$. By symmetry, cooperation is also a strictly dominant strategy for player $3$. 

The single shot utility under the profile $(D,C,C)$ for player 1 is $7 -2\epsilon$. Meanwhile the utilities of players 2 and 3 are both $4.5 + \epsilon$. By Theorem 4.2, this implies 
$\mathbb{P}\left(U^{\infty}_{1} \geq U^{\infty}_{2} = U^{\infty}_{3}\right) = 1$ 
since $\epsilon \leq 7/6$.
Observe that the policy $\rho$ has a much improved margin relative to the trivial policy of repeated defection we first considered. In fact, the margin of policy $\rho$ is $2.5 - \epsilon$, which is the maximum margin achievable by a dominance solvable policy as $\epsilon \to 0$.



\section{Electric vs Petrol Futures}
\label{section: electric petrol futures}

Electric vs Petrol Futures is a three-player repeated game in which each player has to decide on investing in electric or petrol for the upcoming year. The manipulator is an energy company and the other two players are automobile manufacturers who can be considered as payoff maximizing enterprises. The manipulator wants to make as much profit as the other two companies. The actions are $\mathcal{A}_1=\mathcal{A}_2=\mathcal{A}_3 = \{P,E\}$. The manipulator decides whether to invest in batteries for electric cars or petrol cars, while the other two players decide whether the new automobile model they release for the upcoming year will be electric or petrol. The game matrices are:

\begin{align*}
& A^{(1,2)}_0  = 
\begin{bmatrix}
    0.87 & 0\\
    0.8 & 2\\
\end{bmatrix}, \; \;
A^{(1,3)}_0 = 
\begin{bmatrix}
    0.87 & 0\\
    0.8 & 2\\
\end{bmatrix}, \; \;
A^{(2,1)}_0 = \begin{bmatrix}
    2 & 1.5\\
    1.75 & 1.25\\
\end{bmatrix} \\
& A^{(3,1)}_0 = \begin{bmatrix}
    2 & 1.5\\
    1.75 & 1.25\\
\end{bmatrix} \; \;
A^{(2,3)}_0 = \begin{bmatrix}
    0 & 1\\
    0.49 & 0\\
\end{bmatrix} \; \;
A^{(3,2)}_0 = \begin{bmatrix}
    0 & 0.49\\
    1 & 0\\
\end{bmatrix}
\end{align*}

It can be shown that $P$ is the strictly dominant strategy for both P2 and P3. If player 2 and 3 are consistent, then with probability 1, each of their long-run utilities will be at least the utility if they had just played the best constant strategy. Since the best constant strategy is their strictly dominant strategy, they will make utility at least $1.75$. The manipulator can make at most $1.74$ in utility in the long run. No matter what player 1 does, she will lose in the long-run against player 2 and player 3 if they are consistent.\\

\subsection{Winning Strategy for a Manipulator}
It is possible for a manipulator to win this game by following a strategy that is a conventional winning dominance solvable type-\RN{1} policy by enabling convergence to the strategy profile $(\textbf{e}_1,\textbf{e}_1,\textbf{e}_1)$ which stands for $(E,E,E)$.
For some $3/12 < \epsilon < 11/12$,  set  
\begin{equation*}
    \hat{A} = 
    \begin{bmatrix}
        2 - \epsilon & 1.5 + \epsilon \\
        1.75 - \epsilon & 1.25 + \epsilon
    \end{bmatrix}
\end{equation*}

Let player $1$ adopt the policy $\rho_{t} = \left(\textbf{e}_1, \hat{A}, \hat{A}\right)$.
For the profile $(\textbf{e}_1,\textbf{e}_1,\textbf{e}_1)$ the single shot utility for P1 is the the payoff from P2 and P3 which is $2 + 2 = 4$ minus the cost for changing the matrices, which is $2\epsilon$ for a total of $4 - 2\epsilon$. The single-shot utility for P2 is the payoff she gets from P1 which is $1.25 + \epsilon$ and the payoff she gets from P3 which is $0$ for a total of $1.25 + \epsilon$. By symmetry, the payoff to P3 is also $1.25 + \epsilon$. 
By Theorem \ref{thm1} this implies ${\mathbb{P}\left(U_1^{\infty} \geq U_2^{\infty} \text{ and } U_1^{\infty} \geq U_3^{\infty}\right)} = 1$ since $\epsilon < 11/12$.
\section{Battle of the Buddies}
\label{section: battle of budies}

Battle of the Buddies (BoB) is a three-player coordination game which is a generalization of Battle of the Sexes. There are three events, and player $i$ strictly prefers going to event $i$ with both her buddies, over all other outcomes. The game has three actions $\mathcal{A} = \{1,2,3\}$ each of which correspond to going to a particular event. The game matrices are defined as follows:

\[A_0^{(1,2)} = \begin{bmatrix}
    3 & 0 & 0\\
    0 & 2 & 0\\
    0 & 0 & 1
   
\end{bmatrix} \; \; 
A_0^{(1,3)} = \begin{bmatrix}
    3 & 0 & 0\\
    0 & 1 & 0\\
    0 & 0 & 2
   
\end{bmatrix} \; \; 
A_0^{(2,1)} = \begin{bmatrix}
    2 & 0 & 0\\
    0 & 3 & 0\\
    0 & 0 & 1
   
\end{bmatrix} \]
\[ 
A_0^{(2,3)} = \begin{bmatrix}
    1 & 0 & 0\\
    0 & 3 & 0\\
    0 & 0 & 2
   
\end{bmatrix} \; \;
A_0^{(3,1)} = \begin{bmatrix}
    2 & 0 & 0\\
    0 & 1 & 0\\
    0 & 0 & 3
   
\end{bmatrix} \; \;
A_0^{(3,2)} = \begin{bmatrix}
    1 & 0 & 0\\
    0 & 2 & 0\\
    0 & 0 & 3
   
\end{bmatrix}\]

The game is repeated for $T$ rounds. In the following we consider a game in which the original matrices of the game are the game matrices of Battle of the Buddies, but with a manipulator that may change the matrices of P2 and P3. 

\subsection{Winning Strategy for a Manipulator}

It is possible for a manipulator to win BoB by following a winning dominance solvable type-\RN{2} policy. We assume that player 2 is consistent and that player 3 is persistent.
Note that the most advantageous position to any agent in the game is, a position in which the other two agents co-ordinate with her by going with her to her favourite event.
For player one the most advantageous strategy profile is $(\textbf{e}_1,\textbf{e}_1,\textbf{e}_1)$ where she obtains her maximum possible payoff of $\|A_0^{(1,2)}\|_\infty + \|A_0^{(1,3)}\|_\infty = 3+3 = 6$.
For some $0 < \epsilon < 1$ set
\begin{equation*}
    \hat{A} = 
        \begin{bmatrix}
        2.5 + \epsilon & -0.5 & 0\\
        0 & 3 & 0\\
        0 & 0 & 1\\
        \end{bmatrix}    
\end{equation*}

Let player $1$ adopt the policy $\rho_{t} = \left(\textbf{e}_1, \hat{A}, A^{(3,2)}_0\right)$.
It can be shown that this renders $\textbf{e}_1$ as the strictly dominant strategy against the strategy $\textbf{e}_1$ of P1 and any strategy of P3. Further it renders $\textbf{e}_1$ as the strictly dominant strategy for player 3 against the strategy $\textbf{e}_1$ of player 1 and the strategy $\textbf{e}_1$ of player 2.
The single shot utility under the profile $(\textbf{e}_1,\textbf{e}_1,\textbf{e}_1)$ for P1 is the revenue of $6$, minus the cost for changing $A^{(2,1)}_0$ which is $0.5 + \epsilon$. The total is $6 - (0.5 + \epsilon) = 5.5 - \epsilon$. The single-shot utility of P2 is the payoff she gets from P1 which is $2.5 + \epsilon$, plus her payoff from P3 which is $1$ for a total of $3.5 + \epsilon$. The single-shot utility of P3 is $2+1 = 3$.
By Theorem \ref{thm1} this implies ${\mathbb{P}\left(U_1^{\infty} \geq U_2^{\infty} \text{ and } U_1^{\infty} \geq U_3^{\infty}\right)} = 1$ since $\epsilon < 1$.
The fraction of the cost it takes to converge to $(\textbf{e}_{1},\textbf{e}_1,\textbf{e}_1)$ over the revenue for converging to this profile is $\frac{\sum_{(i,j) \in P}{\|A^{(i,j)}_1-A_0^{(i,j)}\|_{\infty}}}{\textbf{e}_{1}^TA^{(1,2)}_t\textbf{e}_{1} + \textbf{e}_{1}^TA^{(1,3)}_t\textbf{e}_{1}} =  \frac{0.5 + \epsilon}{6} \approx 8.33\%$.

The above solution for BoB is optimal in both minimizing the change in cost for modification of the matrices, and maximizing the margin of winning. But if we consider the objective of acting adversarially against player 2 and 3 by minimizing the objective function $f(V) = v_{2} + v_{3}$ (sum of the single-shot utilities of the other two players), then this solution has a sub-optimal objective value of $v_{2} + v_{3} = 4 + \epsilon$.
We now highlight a different solution that is optimal in the sense of maximally reducing the payoff of player 2 and 3 whilst winning the game. 
The strategy profile to converge to is $(\textbf{e}_1,\textbf{e}_1,\textbf{e}_1)$. 
To enable profitable convergence to this strategy profile the manipulator can play the following procedure using a winning dominance solvable type-\RN{2} policy:
For some $0 < \epsilon < 1$ set

\begin{equation*}
    \hat{A} = \begin{bmatrix}
    \epsilon & -3 & -3\\
    0 & 3 & 0\\
    0 & 0 & 1\\
\end{bmatrix} \;\;\;\;\;
    \tilde{A} = 
    \begin{bmatrix}
        \epsilon & -2+\epsilon & -2+\epsilon\\
        0 & 1 & 0\\
        0 & 0 & 3\\
\end{bmatrix}
\end{equation*}


%
%

Let player $1$ adopt the policy $\rho_{t} = \left(\textbf{e}_1, \hat{A}, \tilde{A}\right)$.
It can be shown that this renders $\textbf{e}_1$ as the strictly dominant strategy for player 2 against the strategy $\textbf{e}_1$ of P1 and any strategy of P3, and it renders $\textbf{e}_1$ as the dominant strategy for P3 against the strategy $\textbf{e}_1$ of P1 and $\textbf{e}_1$ of P2.
The single shot utility of P1 is the revenue $6$, minus the cost for modifying $A^{(2,1)}_0$ which is $3$ and the cost for modifying $A^{(3,1)}_0$ which is $2-\epsilon$. The total is $6 - (5 - \epsilon) = 1 + \epsilon$. The single-shot utility of P2 is the payoff she gets from P1 which is $\epsilon$, plus her payoff from P3 which is $1$ for a total of $1 + \epsilon$. The single-shot utility of P3 is $\epsilon+1 = 3$ which she gets from playing her dominant strategy.
Since all three players have the same single-shot utility under this profile, ${\mathbb{P}\left(U_1^{\infty} \geq U_2^{\infty} \text{ and } U_1^{\infty} \geq U_3^{\infty}\right)} = 1$. 
This solution has an objective value of $v_2+v_3 = 2 + 2\epsilon$. We see that the cost for changing the matrices is $5 - \epsilon$ for this solution versus $0.5 + \epsilon$ for the previous solution. We note that by changing the objective from winning with the least cost of modification to winning while maximally reducing the payoff for other players, we end up with completely different solutions. 
\section{Numerical Results}

In our empirical experiments we have tested winning dominance solvable policies and winning batch coordination policies against players that use standard No-Regret Algorithms. This is because all no-regret players are consistent, by Proposition \ref{prop2}. The algorithms we consider for player 2 and 3 are Multiplicative weights update method, Follow-the-Regularized-Leader and Linear multiplicative weights update. For all the games discussed in the paper, the empirical results match the theory. We have set the parameters of MWU and FTRL to ensure that they have fast convergence rates and the parameters of LMWU to ensure that is has a slow convergence rate.\\

We test each policy in at least $200$ game simulations and each game simulation is run for exactly $100$ rounds. We see that if we run the experiments for a large number of rounds the experiments match the theory exactly. This is because of our strong theoretical guarantees. Therefore we have chosen to run each game for only $100$ rounds. For example in the Electric vs Petrol Futures example, if player 2 uses MWU and player 3 uses FTRL then we see that if the manipulator uses the best constant strategy, her Win-Rate is 0\%. However we see that if she uses a winning dominance solvable policy instead her Win-Rate jumps to 100\%.\\

The efficiency of the algorithms we use to find the policies can be increased quite a bit with a simple trick. It is possible to run the sequence of linear feasibility problems (or linear programs) to solve for a winning dominance solvable policy for three players in $\mathcal{O}(nml(m^3+l^3))$ time which is $\mathcal{O}(n^6)$ when $n=m=l$.\footnote{Linear feasibility problems with $k$ variables can be solved using an algorithm with running time that is $\mathcal{O}(k^3)$ by ~\cite{bertsimas97}} This is done by considering only the entries from a single row as the variables from the matrices $A^{(2,1)}$ and $A^{(3,1)}$ instead of considering all of the entries as variables. That is, we change the matrices $A^{(2,1)}$ and $A^{(3,1)}$ only in row $i^*$ where $\textbf{e}_{i^*}$ is the constant strategy of the manipulator. \\

Each of the experiments where at least one player uses Follow-the-Regularized-Leader is run for $N = 200$ simulations and each of those 200 simulations contain $T= 100$ rounds of play.
If no player uses Follow-the-Regularized-Leader then the experiment is run for $N = 2000$ simulations and each of those simulations contain $T = 100$ rounds of play.\\

We calculate the Win-Rate as \[\frac{\text{Number of Games won by Player 1}}{\text{Total number of Games}}\]

We calculate the margin of a single game where \[U_1(\textbf{x}_t)_{t=1}^{100} > U_2(\textbf{y}_t)_{t=1}^{100} \text{ and } U_1(\textbf{x}_t)_{t=1}^{100} > U_2(\textbf{z}_t)_{t=1}^{100}\] as 
\[\min\Big\{U_1(\textbf{x}_t)_{t=1}^{100} - U_2(\textbf{y}_t)_{t=1}^{100}, U_1(\textbf{x}_t)_{t=1}^{100} - U_2(\textbf{z}_t)_{t=1}^{100}\Big\}\] We calculate the Margin of $N$ games (simulations) as the average of the margin of the games where player 1 won. \\

The following subsections display the results of the manipulator playing against different No-Regret Algorithms used by player 2 and 3 while using either the best constant strategy or a winning dominance solvable policy.\\

\subsection{Three-Player Iterated Prisoner's Dilemma}
The following two tables compare the Win-Rate and margin of the best constant strategy vs a winning dominance solvable type-\RN{1} policy in the Three-Player Iterated Prisoner's Dilemma.
\begin{table}[H]
\centering
  \caption{Win-Rate and Margin of Player 1 when she uses the best constant strategy against No-regret Algorithms}
  \label{tab:locations}
  \begin{tabular}{rlll}\toprule
    \textit{Player 2} & \textit{Player 3} & \textit{Win-Rate} & \textit{Margin}\\ \midrule
    MWU & FTRL & 100\% & 0.019\\
    MWU & LMWU & 100\% & 0.03\\
    LMWU & LMWU & 100\% & 0.06\\ \bottomrule
  \end{tabular}
\end{table}

\begin{table}[H]
\centering
  \caption{Win-Rate and Margin of Player 1 when she uses a type-\RN{1} policy against No-regret Algorithms}
  \label{tab:locations}
  \begin{tabular}{rlll}\toprule
    \textit{Player 2} & \textit{Player 3} & \textit{Win-Rate} & \textit{Margin}\\ \midrule
    MWU & FTRL & 100\% & 2.007\\
    MWU & LMWU & 100\% & 1.320\\
    LMWU & LMWU & 100\% & 1.500\\ \bottomrule
  \end{tabular}
\end{table}

We see comparing the Win-Rate of both policies that both of these solutions are equally good. However if the manipulator has an additional objective of winning by a large margin, the type-\RN{1} policy is clearly better than the best constant strategy. We see that within 100 rounds the margin under the type-\RN{1} policy is close to the infinite-horizon margin which is $2.5 - 3\epsilon$ (provided the other two players use No-Regret algorithms with fast convergence rates such as MWU and FTRL in this case). We have chosen $\epsilon = 0.1$ for our experiments, and so the infinite horizon margin is $2.20$ and we see that if we run each simulation for $1000$ rounds, the margin is within $0.01$ of $2.20$. 
\subsection{Social Distancing Game}
The following two tables compare the Win-Rate and margin of the best constant strategy vs a winning dominance solvable type-\RN{1} policy in the Social Distancing Game.
\begin{table}[H]
\centering
  \caption{Win-Rate and Margin of Player 1 when she uses the best constant strategy against No-regret Algorithms}
  \label{tab:locations}
  \begin{tabular}{rlll}\toprule
    \textit{Player 2} & \textit{Player 3} & \textit{Win-Rate} & \textit{Margin}\\ \midrule
    MWU & FTRL & 98.75\% & 2.59\\
    MWU & LMWU & 77.7\% & 1.01\\
    LMWU & LMWU & 99.9\% & 2.72\\ \bottomrule
  \end{tabular}
\end{table}

\begin{table}[H]
\centering
  \caption{Win-Rate and Margin of Player 1 when she uses a type-\RN{1} policy against No-regret Algorithms}
  \label{tab:locations}
  \begin{tabular}{rlll}\toprule
    \textit{Player 2} & \textit{Player 3} & \textit{Win-Rate} & \textit{Margin}\\ \midrule
    MWU & FTRL & 100\% & 5.35\\
    MWU & LMWU & 79\% & 1.09\\
    LMWU & LMWU & 99.8\% & 2.90\\ \bottomrule
  \end{tabular}
\end{table}

We see comparing the Win-Rate of both policies that both of these solutions are quite similar although the type-\RN{1} policy does have a 100\% Win-Rate when playing against MWU and FTRL. We see that the margin under the type-\RN{1} policy is clearly better than the best constant strategy. We see that within 100 rounds the margin under the type-\RN{1} policy is close to the infinite-horizon margin which is $6 - 2\epsilon$ (provided the other two players use No-Regret algorithms with fast convergence rates such as MWU and FTRL in this case). We have chosen $\epsilon = 0.1$ for our experiments, and so the infinite horizon margin is $5.80$ and we see that if we run each simulation for $1000$ rounds, the margin is within $0.01$ of $5.80$.\\

It is important to note that the Win-Rates under the best constant strategy and the type-\RN{1} policy are close purely because of the no-regret algorithms we have chosen. It is easy to construct examples of adversarial no-regret algorithms against which playing a constant strategy (without manipulating the matrices) will never lead to a win. For example, if player 2 and 3 coordinate on playing at 3 o'clock and 9 o'clock respectively for all time, then these are no-regret algorithms against which player 1 can never win by simply playing a constant strategy. 

\subsection{Electric vs Petrol Futures}
The following two tables compare the Win-Rate and margin of the best constant strategy vs a winning dominance solvable type-\RN{1} policy in the Electric vs Petrol Futures Game.
\begin{table}[H]
\centering
  \caption{Win-Rate and Margin of Player 1 when she uses the best constant strategy against No-regret Algorithms}
  \label{tab:locations}
  \begin{tabular}{rlll}\toprule
    \textit{Player 2} & \textit{Player 3} & \textit{Win-Rate} & \textit{Margin}\\ \midrule
    MWU & FTRL & 0\% & 0.0\\
    MWU & LMWU & 0\% & 0.0\\
    LMWU & LMWU & 14\% & 0.0143\\ \bottomrule
  \end{tabular}
\end{table}

\begin{table}[H]
\centering
  \caption{Win-Rate and Margin of Player 1 when she uses a type-\RN{1} policy against No-regret Algorithms}
  \label{tab:locations}
  \begin{tabular}{rlll}\toprule
    \textit{Player 2} & \textit{Player 3} & \textit{Win-Rate} & \textit{Margin}\\ \midrule
    MWU & FTRL & 100\% & 0.264\\
    MWU & LMWU & 88.5\% & 0.1319\\
    LMWU & LMWU & 68\% & 0.114\\ \bottomrule
  \end{tabular}
\end{table}

We had constructed this example to showcase the existence of a game in which manipulation of the game matrices is essential to even stand a chance at winning in the long-run against consistent players. Without this ability, no matter what player 1 does, she will lose in the long-run against player 2 and player 3 if they use no-regret algorithms. This is exactly what we see in our experimental results. Comparing the Win-Rate of both policies, we see that the type-\RN{1} policy is substantially better than the best constant strategy. In particular if player 2 and 3 use MWU and FTRL respectively, then the best constant strategy has a Win-Rate of 0\%, while the type-\RN{1} policy has a Win-Rate of 100\%. When the Win-Rate is 0\%, the margin is also zero. However when player 2 and 3 use LWMU, we see that the margin of the best constant strategy is positive, but the margin of the type-\RN{1} policy is 8 times that of the former policy.

\subsection{Battle of the Buddies}
The following two tables compare the Win-Rate and margin of the best constant strategy vs a winning dominance solvable type-\RN{2} policy in the Battle of the Buddies.
\begin{table}[H]
\centering
  \caption{Win-Rate and Margin of Player 1 when she uses the best constant strategy against No-regret Algorithms}
  \label{tab:locations}
  \begin{tabular}{rlll}\toprule
    \textit{Player 2} & \textit{Player 3} & \textit{Win-Rate} & \textit{Margin}\\ \midrule
    MWU & FTRL & 100\% & 2.949\\
    MWU & LMWU & 100\% & 2.923\\
    LMWU & LMWU & 100\% & 2.89\\ \bottomrule
  \end{tabular}
\end{table}

\begin{table}[H]
\centering
  \caption{Win-Rate and Margin of Player 1 when she uses a type-\RN{2} policy against No-regret Algorithms}
  \label{tab:locations}
  \begin{tabular}{rlll}\toprule
    \textit{Player 2} & \textit{Player 3} & \textit{Win-Rate} & \textit{Margin}\\ \midrule
    MWU & FTRL & 100\% & 1.76\\
    MWU & LMWU & 100\% & 1.74\\
    LMWU & LMWU & 100\% & 1.73\\ \bottomrule
  \end{tabular}
\end{table}

In this example we see that the Win-Rate of the best constant strategy and the type-\RN{2} policy are equally good. However, the margin of the best constant strategy is twice as good. 
This is purely because of the No-Regret algorithms we have chosen. It is easy to construct examples of adversarial No-Regret algorithms against which playing a constant strategy (without manipulating the matrices) will never lead to a win. For example, if player 2 and 3 coordinate on going to event 3 for all time, then these are No-Regret algorithms against which player 1 can never win. But, the type-\RN{2} policy will always win in the long-run. We believe this example showcases the robustness of our solution.

\section{Conclusions and Further Work}
\label{sec:conclusions}

In this paper, we considered a 3-player iterated polymatrix game setting in which our agent is allowed to (slightly) manipulate the underlying game matrices of the other agents for which she pays a manipulation cost, while the other agents are 'consistent'. In our framework, two examples of consistent agents are those that use follow-the-leader or any no-regret algorithm to play the game.
We first proposed a payoff matrix manipulation scheme and policies for our agent that provably guarantees that the utility of any consistent opponent would converge to a value we desire.

Using this theory we developed winning dominance solvable policies and winning batch coordination policies, both of which have strong theoretical guarantees such as tractability and the ability to win in a finite number of rounds almost surely. In addition, we showed that these policies can be found efficiently by solving a sequence of linear feasibility problems. We then considered additional objectives the manipulator may have, such as winning by the largest margin or whilst seeing a large improvement relative to the cost of modifying the payoff matrices. We then considered a socially good objective different from winning, namely maximization of the egalitarian social welfare. We showed that our framework could be extended to capture such objectives via linear objective functions. After this, we considered a social distancing game and showed that, by making only infinitesimal changes to the payoff matrices, the manipulator can maximize her payoff i.e. maximize her distance from the other players. The manipulator can also guide the utilities of all players to converge to a socially optimal solution.

One can observe that the results presented in this paper extend trivially to generic $n$-player games where the manipulator can change payoff tensors of the other players for an infinity-norm cost. A manipulator can achieve the same theoretical guarantees we have presented, in the more general setting, by solving a sequence of linear feasibility problems that is similar to the one presented in this paper. 
In a generic $n$-player setting, where each player has $k$ actions, the generalization of the our proposed winning policies can be found using an algorithm with running time that is $\mathcal{O}(k^{4n-3})$.

In many practical settings, a manipulator may have less control over the underlying system than what has been assumed in our setting. With this in mind, a potential direction for future work would be to consider a setting in which the manipulator is further constrained in the way she can manipulate the underlying game. Similarly, there are many (context-dependent) goals and objectives that the manipulator may have that were not tackled in this paper, which form interesting topics for future work. Lastly, in many real world scenarios, the manipulator may only have partial knowledge of the underlying game. It would be interesting to see if this uncertainty could be integrated into the framework we have devised. For example in the Electric vs Petrol Futures example, if player 2 uses MWU and player 3 uses FTRL then we see that if the manipulator uses the best constant strategy, her Win-Rate is 0\%. However we see that if she uses a winning dominance solvable policy instead her Win-Rate jumps to 100\%.

\bibliographystyle{unsrt}  
\bibliography{references}  

\newpage
\appendix

\section{Proofs}

We denote the expected utility of player $i$ by \begin{equation*}
    \widehat{U}_i(\textbf{x}_t,\textbf{y}_t,\textbf{z}_t)_{t=1}^{T} = \mathbb{E}[U_i(\textbf{x}_t,\textbf{y}_t,\textbf{z}_t)_{t=1}^{T}]
\end{equation*}

We further denote the maximum and minimum single-shot utility that can be attained by player $i$ as $U_i^{\text{max}} \text{ and }  U_i^{\text{min}}$.

\begin{lemma}
\label{lem1}
Assume the manipulator uses a winning dominance solvable policy. Within $T$ rounds of the game, let $T_2$ be the number of rounds player 2 plays her strictly dominant action and $T_3$ be the number of rounds player 3 plays her strictly dominant action. If $ \; \mathbb{P}\big(\lim_{T \rightarrow \infty}{\frac{T_2}{T} = 1}\big) = 1$ and $\mathbb{P}\big(\lim_{T \rightarrow \infty}{\frac{T_3}{T} = 1}\big) = 1$ then 
\begin{align*}
& \mathbb{P}\Big(U_1(\textbf{x}_t,\textbf{y}_t,\textbf{z}_t)_{t=1}^{\infty} \geq U_2(\textbf{x}_t,\textbf{y}_t,\textbf{z}_t)_{t=1}^{\infty} \text{ and } U_1(\textbf{x}_t,\textbf{y}_t,\textbf{z}_t)_{t=1}^{\infty} \geq U_3(\textbf{x}_t,\textbf{y}_t,\textbf{z}_t)_{t=1}^{\infty} \Big) = 1
\end{align*}
\end{lemma}
\begin{proof}
If $T_1$ was the number of rounds in which both player 2 and 3 play their respective strictly dominant actions,\\ then $\lim_{T \rightarrow \infty}{\frac{T_2}{T} = 1}$ and $\lim_{T \rightarrow \infty}{\frac{T_3}{T} = 1}$ implies $\lim_{T \rightarrow \infty}{\frac{T_1}{T} = 1}$
Hence,
\begin{align*}
& \mathbb{P}\big(\lim_{T \rightarrow \infty}{\frac{T_1}{T} = 1}\big) \geq \mathbb{P}\big(\lim_{T \rightarrow \infty}{\frac{T_2}{T} = 1} \text{ and } \lim_{T \rightarrow \infty}{\frac{T_3}{T} = 1}\big) = 1
\end{align*}
We now give upper and lower bounds for the utility for of every player $p \in \mathcal{N}$ within $T$ rounds. We then show that the long-run utility converges to $\widehat{U}_p(\textbf{e}_{i^*},\textbf{e}_{j^*},\textbf{e}_{k^*})$ using a sandwiching argument. For player $p \in \mathcal{N}$,

\begin{align*}
& \frac{T_1}{T}[\widehat{U}_p(\textbf{e}_{i^*},\textbf{e}_{j^*},\textbf{e}_{k^*}) - U_p^{\text{max}}] + U_p^{\text{max}} \\
& = \frac{T_1}{T}\widehat{U}_p(\textbf{e}_{i^*},\textbf{e}_{j^*},\textbf{e}_{k^*}) + \frac{T-T_1}{T}U_p^{\text{max}}  \\
& \geq \frac{ T_1\widehat{U}_p(\textbf{e}_{i^*},\textbf{e}_{j^*},\textbf{e}_{k^*}) + \sum_{1 \leq t \leq T : \textbf{y}_t \neq \textbf{e}_{j^*} \textup{ or } \textbf{z}_t \neq \textbf{e}_{k^*}}{  U_p(\textbf{e}_{i^*},\textbf{y}_t,\textbf{z}_t)}}{T}  \\
& = U_p(\textbf{x}_t,\textbf{y}_t,\textbf{z}_t)_{t=1}^{T}   \\
& = \frac{ T_1\widehat{U}_p(\textbf{e}_{i^*},\textbf{e}_{j^*},\textbf{e}_{k^*}) + \sum_{1 \leq t \leq T : \textbf{y}_t \neq \textbf{e}_{j^*} \textup{ or } \textbf{z}_t \neq \textbf{e}_{k^*}}{  U_p(\textbf{e}_{i^*},\textbf{y}_t,\textbf{z}_t)}}{T}  \\
& \geq \frac{T_1}{T}\widehat{U}_p(\textbf{e}_{i^*},\textbf{e}_{j^*},\textbf{e}_{k^*}) + \frac{T-T_1}{T}U_p^{\text{min}}  \\
& \geq \frac{T_1}{T}[\widehat{U}_p(\textbf{e}_{i^*},\textbf{e}_{j^*},\textbf{e}_{k^*}) - U_p^{\text{min}}] + U_p^{\text{min}} 
\end{align*}
Assume $\lim_{T \rightarrow \infty}{\frac{T_1}{T}} = 1$. Then by sandwich theorem we have,
\begin{align*}
& \lim_{T \rightarrow \infty}U_p(\textbf{x}_t,\textbf{y}_t,\textbf{z}_t)_{t=1}^{T} = \widehat{U}_p(\textbf{e}_{i^*},\textbf{e}_{j^*},\textbf{e}_{k^*})
\end{align*}
This implies that for player $p$,
\begin{align*}
& \mathbb{P}\left(U_p(\textbf{x}_t,\textbf{y}_t,\textbf{z}_t)_{t=1}^{\infty} = \widehat{U}_p(\textbf{e}_{i^*},\textbf{e}_{j^*},\textbf{e}_{k^*})\right) \geq \mathbb{P}\Big(\lim_{T \rightarrow \infty}{\frac{T_1}{T} = 1}\Big) = 1
\end{align*}
Since the strategy profile in question is a winning dominance solvable policy,
\begin{align*}
& \widehat{U}_1(\textbf{e}_{i^*},\textbf{e}_{j^*},\textbf{e}_{k^*}) \geq \widehat{U}_2(\textbf{e}_{i^*},\textbf{e}_{j^*},\textbf{e}_{k^*}) \text{ and } \widehat{U}_1(\textbf{e}_{i^*},\textbf{e}_{j^*},\textbf{e}_{k^*}) \geq \widehat{U}_3(\textbf{e}_{i^*},\textbf{e}_{j^*},\textbf{e}_{k^*})    
\end{align*}
which implies 
\begin{align*}
& \mathbb{P}\Big(U_1(\textbf{x}_t,\textbf{y}_t,\textbf{z}_t)_{t=1}^{\infty} \geq U_2(\textbf{x}_t,\textbf{y}_t,\textbf{z}_t)_{t=1}^{\infty} \text{ and } U_1(\textbf{x}_t,\textbf{y}_t,\textbf{z}_t)_{t=1}^{\infty} \geq U_3(\textbf{x}_t,\textbf{y}_t,\textbf{z}_t)_{t=1}^{\infty} \Big)  \\
& \geq \mathbb{P}\Big(U_1(\textbf{x}_t,\textbf{y}_t,\textbf{z}_t)_{t=1}^{\infty} = \widehat{U}_1(\textbf{e}_{i^*},\textbf{e}_{j^*},\textbf{e}_{k^*}), U_2(\textbf{x}_t,\textbf{y}_t,\textbf{z}_t)_{t=1}^{\infty} = \widehat{U}_2(\textbf{e}_{i^*},\textbf{e}_{j^*},\textbf{e}_{k^*}) \text{ and } \\
& U_3(\textbf{x}_t,\textbf{y}_t,\textbf{z}_t)_{t=1}^{\infty} = \widehat{U}_3(\textbf{e}_{i^*},\textbf{e}_{j^*},\textbf{e}_{k^*})\Big) = 1
\end{align*}
\end{proof}

\begin{proof}[Proof of Theorem \ref{thm1}]
\label{prf1}
Since player 1 plays a winning dominance solvable type-\RN{1} policy, player 1 plays $\textbf{e}_{i^*}$ every round. Further player 2 and player 3 each have strictly dominant strategies $\textbf{e}_{j^*}$ and $\textbf{e}_{k^*}$ against the fixed strategy $\textbf{e}_{i^*}$ of player 1. Let $T_2$ be the number of rounds within $T$ rounds that player 2 plays action $j^*$. Similarly let $T_3$ be the number of rounds within $T$ rounds that player 3 plays action $k^*$. Since both player 2 and 3 are consistent, this implies $\mathbb{P}\big(\lim_{T \rightarrow \infty}{\frac{T_2}{T} = 1}\big) = 1$ and $\mathbb{P}\big(\lim_{T \rightarrow \infty}{\frac{T_3}{T} = 1}\big) = 1$.\\   

By Lemma \ref{lem1}, 
\begin{align*}
& \mathbb{P}\Big(U_1(\textbf{x}_t,\textbf{y}_t,\textbf{z}_t)_{t=1}^{\infty} \geq U_2(\textbf{x}_t,\textbf{y}_t,\textbf{z}_t)_{t=1}^{\infty} \text{ and }  U_1(\textbf{x}_t,\textbf{y}_t,\textbf{z}_t)_{t=1}^{\infty} \geq U_3(\textbf{x}_t,\textbf{y}_t,\textbf{z}_t)_{t=1}^{\infty} \Big) = 1
\end{align*}
\end{proof}
For the following proof we assume the single-shot utilities of each player is bounded between $-1$ and $1$.
\begin{proof}[Proof of Theorem \ref{thm2}]
Since player 1 plays a winning dominance solvable type-\RN{2} policy, player 1 plays $\textbf{e}_{i^*}$ every round. Let $T_2$ be the number of rounds player 2 plays her strictly dominant action $j^*$ within $T$ rounds. Then $\mathbb{P}\Big(\lim_{T \rightarrow \infty}{\frac{T_2}{T}} = 1\Big) = 1$. That is, the event $\mathcal{E}$, "For any $\epsilon>0$, there exists a $T_0(\epsilon)$ such that for all $T > T_0(\epsilon)$, $\frac{T_2}{T} > 1 - \epsilon$" holds almost surely.\\

Set \[\epsilon_0 = \frac{\widehat{U}_3(\textbf{e}_{i^*},\textbf{e}_{j^*},\textbf{e}_{k^*})-\widehat{U}_3(\textbf{e}_{i^*},\textbf{e}_{j^*},\textbf{e}_{k^{**}})}{4}\] where $k^{**}$ is the second best response to the strategies $\textbf{e}_{i^*}$ and $\textbf{e}_{j^*}$ of player 1 and 2. Now assuming the event $\mathcal{E}$ holds and that $T>T_0(\epsilon_0)$, we bound the utility of player 3, and show that $k^*$ eventually becomes the best action in hindsight for player 3.  

\begin{align*}
& 2\epsilon + \widehat{U}_3(\textbf{e}_{i^*},\textbf{e}_{j^*},\textbf{e}_k)\\
& \geq \epsilon[U_3^{\text{max}} - \widehat{U}_3(\textbf{e}_{i^*},\textbf{e}_{j^*},\textbf{e}_k) ] + \widehat{U}_3(\textbf{e}_{i^*},\textbf{e}_{j^*},\textbf{e}_k)\\
& > \frac{T-T_2}{T}[U_3^{\text{max}} - \widehat{U}_3(\textbf{e}_{i^*},\textbf{e}_{j^*},\textbf{e}_k) ] + \widehat{U}_3(\textbf{e}_{i^*},\textbf{e}_{j^*},\textbf{e}_k) \\
& = \frac{T_2}{T}\widehat{U}_3(\textbf{e}_{i^*},\textbf{e}_{j^*},\textbf{e}_k) + \frac{T-T_2}{T}U_3^{\text{max}}  \\
& \geq \frac{ T_2\widehat{U}_3(\textbf{e}_{i^*},\textbf{e}_{j^*},\textbf{e}_k) + \sum_{1 \leq t \leq T : \textbf{y}_t \neq \textbf{e}_{j^*}}{  U_3(\textbf{e}_{i^*},\textbf{y}_t,\textbf{e}_k)}}{T}  \\
& = U_3(\textbf{x}_t,\textbf{y}_t,\textbf{e}_k)_{t=1}^{T}   \\
& = \frac{ T_2\widehat{U}_3(\textbf{e}_{i^*},\textbf{e}_{j^*},\textbf{e}_k) + \sum_{1 \leq t \leq T : \textbf{y}_t \neq \textbf{e}_{j^*}}{  U_3(\textbf{e}_{i^*},\textbf{y}_t,\textbf{e}_k)}}{T}  \\
& \geq \frac{T_2}{T}\widehat{U}_3(\textbf{e}_{i^*},\textbf{e}_{j^*},\textbf{e}_k) + \frac{T-T_2}{T}U_3^{\text{min}}  \\
& = \frac{T_2}{T}[\widehat{U}_3(\textbf{e}_{i^*},\textbf{e}_{j^*},\textbf{e}_k) - U_3^{\text{min}}] + U_3^{\text{min}} \\
& > (1-\epsilon)[\widehat{U}_3(\textbf{e}_{i^*},\textbf{e}_{j^*},\textbf{e}_k) - U_3^{\text{min}}] + U_3^{\text{min}} \\
& = \widehat{U}_3(\textbf{e}_{i^*},\textbf{e}_{j^*},\textbf{e}_k) -\epsilon[\widehat{U}_3(\textbf{e}_{i^*},\textbf{e}_{j^*},\textbf{e}_k) - U_3^{\text{min}}] \\
& \geq \widehat{U}_3(\textbf{e}_{i^*},\textbf{e}_{j^*},\textbf{e}_k) -2\epsilon
\end{align*}

It can be shown that for our choice of $\epsilon_0$, this makes $k^*$ the unique best response in hindsight from round $T_0(\epsilon_0) + 1$ onwards. This implies \[\mathbb{P}\Big(\textbf{e}_{k^*} = \argmax_{\textbf{z} \in \Delta_l}{U_3(\textbf{e}_{i^*},\textbf{y}_{t},\textbf{z})_{t=1}^{T}} \text{ eventually}\Big)=1\] Suppose $T_3$ is the number of rounds player 3 plays action $k^*$ within $T$ rounds. Now since player 3 is persistent, $\mathbb{P}\Big(\lim_{T\rightarrow\infty}{\frac{T_3}{T}} = 1\Big)=1$. In conclusion, by Lemma \ref{lem1}, 
\begin{align*}
& \mathbb{P}\Big(U_1(\textbf{x}_t,\textbf{y}_t,\textbf{z}_t)_{t=1}^{\infty} \geq U_2(\textbf{x}_t,\textbf{y}_t,\textbf{z}_t)_{t=1}^{\infty} \text{ and }  U_1(\textbf{x}_t,\textbf{y}_t,\textbf{z}_t)_{t=1}^{\infty} \geq U_3(\textbf{x}_t,\textbf{y}_t,\textbf{z}_t)_{t=1}^{\infty} \Big) = 1
\end{align*}
\end{proof}
\begin{lemma}
\label{lem5}
Winning dominance solvable type-\RN{1} policies can be found in polynomial time if they exist.
\end{lemma}
\begin{proof}
For $(\textbf{e}_{i^*},\textbf{e}_{j^*},\textbf{e}_{k^*})$ with $i^* \in [n], j^* \in [m], k^* \in [l]$, we obtain the winning strategy as a feasible solution of one of the $n \times m \times l$ linear feasibility problems with the set of variables\\ $V = \{d_2,d_3,v_{2,k},v_{3,j}, A^{(2,1)}_{i,j}, A^{(3,1)}_{i,k} \; :\; i \in [n], \; j \in [m], \; k \in [l]\}$.\\
The linear constraints of each linear feasibility problem are:
\begin{align*}
& {|A^{(2,1)}(i,j) - A_0^{(2,1)}(i,j)|}{\leq d_2 \;\; \forall (i,j) \in [n]\times[m]}{}\\
& {|A^{(3,1)}(i,k) - A_0^{(3,1)}(i,k)|}{\leq d_3 \;\; \forall (i,k) \in [n]\times[l]}{}\\
& {A^{(2,1)}(i^*,j) + A^{(2,3)}_0(j,k)}{=  v_{2,k} \text{ for} \; \; j=j^* \; \; k \in [l]}{}\\
& {A^{(2,1)}(i^*,j) + A^{(2,3)}_0(j,k)}{\leq  v_{2,k} - \epsilon \text{ for} \; \; j\neq j^* \; \; k \in [l]}{}\\
& {A^{(3,1)}(i^*,k) + A^{(3,2)}_0(j,k)}{=  v_{3,j} \text{ for} \; \; k= k^* \; \; j \in [m]}{}\\
& {A^{(3,1)}(i^*,k) + A^{(3,2)}_0(j,k)}{\leq  v_{3,j} - \epsilon \text{ for} \; \; k\neq k^* \; \; j \in [m]}{}\\
& {v_{2,k^*}}{ \leq A^{(1,2)}_{i^*,j^*} + A^{(1,3)}_{i^*,k^*} - d_2 - d_3}{ }\\
& {v_{3,j^*}}{ \leq A^{(1,2)}_{i^*,j^*} + A^{(1,3)}_{i^*,k^*} - d_2 - d_3}{ }\\
\end{align*}
The first set of linear constraints are equivalent to $\|A^{(2,1)} - A_0^{(2,1)}\|_{\infty} \leq d_2$ and the second set is equivalent to $\|A^{(3,1)} - A_0^{(3,1)}\|_{\infty} \leq d_3$. If we let $\textbf{x}  = \textbf{e}_{i^*}$ then the third and fourth set of constraints correspond to system \eqref{eq:1}, and the fifth and sixth set of constraints correspond to system \eqref{eq:2}. Note that $v_{2,k^*} = \widehat{U}_2(\textbf{e}_{i^*},\textbf{e}_{j^*},\textbf{e}_{k^*})$ and $v_{3,j^*} = \widehat{U}_3(\textbf{e}_{i^*},\textbf{e}_{j^*},\textbf{e}_{k^*})$ are the respective payoffs of the strategy profile that players converge to. Therefore the second last and last constraint imply that 
\begin{align*}
& \widehat{U}_2(\textbf{e}_{i^*},\textbf{e}_{j^*},\textbf{e}_{k^*}) ,\; \widehat{U}_3(\textbf{e}_{i^*},\textbf{e}_{j^*},\textbf{e}_{k^*}) = v_{2,k^*} ,\; v_{3,j^*} \leq A^{(1,2)}_{i^*,j^*} + A^{(1,3)}_{i^*,k^*} - d_2 - d_3 \\
& \leq A^{(1,2)}_{i^*,j^*} + A^{(1,3)}_{i^*,k^*} - \|A^{(2,1)} - A_0^{(2,1)}\|_{\infty} - \|A^{(3,1)} - A_0^{(3,1)}\|_{\infty} = \widehat{U}_1(\textbf{e}_{i^*},\textbf{e}_{j^*},\textbf{e}_{k^*})
\end{align*}
That is, the final two set of constraints ensure that the strategy profile and matrices give rise to a winning policy for player 1. Therefore a strategy that satisfies the linear constraints is a winning dominance solvable policy.\\
The run-time for executing the sequence of linear feasibility problems is $\mathcal{O}(nml(n^3m^3+l^3n^3))$ and therefore the existence of an algorithm that has running time polynomial in the number of actions of the players is proven.
\end{proof}
\begin{lemma}
\label{lem2}
Winning dominance solvable Type-\RN{2} policies can be found in polynomial time if they exist.
\end{lemma}
\begin{proof}
For $(\textbf{e}_{i^*},\textbf{e}_{j^*},\textbf{e}_{k^*})$ with $i^* \in [n], j^* \in [m], k^* \in [l]$, we obtain the winning strategy as a feasible solution of one of the $n \times m \times l$ linear feasibility problems with the set of variables\\ $V = \{d_2,d_3,v_{2,k},v_3, A^{(2,1)}_{i,j}, A^{(3,1)}_{i,k} \; :\; i \in [n], \; j \in [m], \; k \in [l]\}$:\\
The linear constraints of each linear feasibility problem are:
\begin{align*}
& {|A^{(2,1)}(i,j) - A_0^{(2,1)}(i,j)|}{\leq d_2 \;\; \forall (i,j) \in [n]\times[m]}{}\\
& {|A^{(3,1)}(i,k) - A_0^{(3,1)}(i,k)|}{\leq d_3 \;\; \forall (i,k) \in [n]\times[l]}{}\\
& {A^{(2,1)}(i^*,j) + A^{(2,3)}_0(j,k)}{=  v_{2,k} \text{ for} \; \; j=j^* \; \; k \in [l]}{}\\
& {A^{(2,1)}(i^*,j) + A^{(2,3)}_0(j,k)}{\leq  v_{2,k} - \epsilon \text{ for} \; \; j\neq j^* \; \; k \in [l]}{}\\
& {A^{(3,1)}(i^*,k) + A^{(3,2)}_0(j^*,k)}{=  v_{3} \text{ for} \; \; k= k^*}{}\\
& {A^{(3,1)}(i^*,k) + A^{(3,2)}_0(j^*,k)}{\leq  v_{3} - \epsilon \text{ for} \; \; k\neq k^*}{}\\
& {v_{2,k^*}}{ \leq A^{(1,2)}_{i^*,j^*} + A^{(1,3)}_{i^*,k^*} - d_2 - d_3}{ }\\
& {v_{3,j^*}}{ \leq A^{(1,2)}_{i^*,j^*} + A^{(1,3)}_{i^*,k^*} - d_2 - d_3}{ }\\
\end{align*}
The only difference in analysis from the linear constraints from Lemma \ref{lem5} and this one is that the 5th and 6th set of constraints in this lemma correspond to system \eqref{eq:3} instead of system \eqref{eq:2}. The rest of the analysis is identical.
\end{proof}
\begin{proof}[Proof of Theorem \ref{thm3}]
The proof is an immediate consequence of Lemma \ref{lem5} and Lemma \ref{lem2}.
\end{proof}

For the following proof we assume the single-shot utilities of each player is bounded between $-1$ and $1$.
\begin{proof}[Proof of Theorem \ref{thm4}]
Let 
\begin{align*}
& a_1 = \widehat{U}_1(\textbf{e}_{i_2},\textbf{e}_{j_2},\textbf{e}_{k_2}) \;\;\; a_2 = \widehat{U}_1(\textbf{e}_{i_3},\textbf{e}_{j_3},\textbf{e}_{k_3}) \\
& b_1 = \widehat{U}_2(\textbf{e}_{i_2},\textbf{e}_{j_2},\textbf{e}_{k_2}) \;\;\; b_2 = \widehat{U}_2(\textbf{e}_{i_3},\textbf{e}_{j_3},\textbf{e}_{k_3}) \\
& c_1 = \widehat{U}_3(\textbf{e}_{i_2},\textbf{e}_{j_2},\textbf{e}_{k_2}) \;\;\; c_2 = \widehat{U}_3(\textbf{e}_{i_3},\textbf{e}_{j_3},\textbf{e}_{k_3}) 
\end{align*}

We focus our attention to the first half of the game (rounds $1,...,T$) and derive analogous results for the second half (rounds $T+1,...,2T$) using a symmetry argument. First, player 2 and 3 are no-regret, which implies both of them are consistent (by Proposition \ref{prop2}). Note that when player 1 plays a dominance solvable type-\RN{1} policy in the first half, the fraction of time player 2 and 3 play their respective strictly dominant strategies converges to 1 (almost surely) since they are consistent. Let $T_1$ be the number of rounds where player 2 plays action $j_2$ and player 3 plays action $k_2$ within $T$ rounds. Then $\mathbb{P}(\lim_{T \rightarrow \infty}\frac{T_1}{T} = 1) = 1$.\\ So $\frac{T_1}{T} \overset{a.s.}{\to} 1$. That is, the event $\mathcal{E}$, "For any $\epsilon>0$, there exists a $T_0(\epsilon)$ such that for all $T > T_0(\epsilon)$, $\frac{T_1}{T} > 1 - \epsilon$" holds almost surely.\\

Set 
\begin{align*}
& \epsilon_0 = \frac{1}{6}\Big(a_1+a_2 - \max\big\{b_1+b_2, c_1+c_2\big\}\Big)
\end{align*}

Now assuming the event $\mathcal{E}$ holds and that $T>T_0(\epsilon_0)$, we bound the utility of player 1.
\begin{align*}
& a_1 + 2\epsilon_0 \\
& \geq \epsilon_0[U_1^{\text{max}} + a_1 ] + a_1 \\ 
& > \frac{T - T_1}{T}[U_1^{\text{max}} +a_1 ] + a_1 \\ 
& = \frac{T_1}{T}a_1 + \frac{T-T_1}{T}U_1^{\text{max}}  \\
& \geq \frac{ T_1a_1 + \sum_{1 \leq t \leq T : \textbf{y}_t \neq \textbf{e}_{j^*} \textup{ or } \textbf{z}_t \neq \textbf{e}_{k^*}}{  U_1(\textbf{e}_{i^*},\textbf{y}_t,\textbf{z}_t)}}{T}  \\
& = U_1(\textbf{x}_t,\textbf{y}_t,\textbf{z}_t)_{t=1}^{T}   \\
& = \frac{ T_1a_1 + \sum_{1 \leq t \leq T : \textbf{y}_t \neq \textbf{e}_{j^*} \textup{ or } \textbf{z}_t \neq \textbf{e}_{k^*}}{  U_1(\textbf{e}_{i^*},\textbf{y}_t,\textbf{z}_t)}}{T}  \\
& \geq \frac{T_1}{T}a_1 + \frac{T-T_1}{T}U_1^{\text{min}}  \\
& \geq \frac{T_1}{T}[a_1 - U_1^{\text{min}}] + U_1^{\text{min}} \\
& > a_1 - \epsilon_0(a_1 - U_1^{\text{min}})   \\
& > a_1 - 2\epsilon_0  \\
\end{align*}
which renders 
\[2\epsilon_0 > \big|U_1(\textbf{x}_t,\textbf{y}_t,\textbf{z}_t)_{t=1}^{T} - a_1\big| \]
We can derive bounds on the utility of player 2 and 3 in an identical fashion. This would give
\[2\epsilon_0 > \big|U_2(\textbf{x}_t,\textbf{y}_t,\textbf{z}_t)_{t=1}^{T} - b_1\big| \]
\[2\epsilon_0 > \big|U_3(\textbf{x}_t,\textbf{y}_t,\textbf{z}_t)_{t=1}^{T} -c_1\big| \]
Using a symmetric argument, we obtain the following bounds for the second half of the game:
\[2\epsilon_0 > \big|U_1(\textbf{x}_t,\textbf{y}_t,\textbf{z}_t)_{t=T}^{2T} - a_2\big| \]
\[2\epsilon_0 > \big|U_2(\textbf{x}_t,\textbf{y}_t,\textbf{z}_t)_{t=T}^{2T} - b_2\big| \]
\[2\epsilon_0 > \big|U_3(\textbf{x}_t,\textbf{y}_t,\textbf{z}_t)_{t=T}^{2T} - c_2\big| \]
Combining the bounds for the first half and second half, we obtain
\[2\epsilon_0 > \big|U_1(\textbf{x}_t,\textbf{y}_t,\textbf{z}_t)_{t=1}^{2T} - \frac{a_1+a_2}{2}\big| \]
\[2\epsilon_0 > \big|U_2(\textbf{x}_t,\textbf{y}_t,\textbf{z}_t)_{t=1}^{2T} - \frac{b_1+b_2}{2}\big| \]
and \[2\epsilon_0 > \big|U_3(\textbf{x}_t,\textbf{y}_t,\textbf{z}_t)_{t=1}^{2T} - \frac{c_1+c_2}{2}\big| \]
It can be shown that our choice of $\epsilon_0$ gives us 
\begin{align*}
& U_1(\textbf{x}_t,\textbf{y}_t,\textbf{z}_t)_{t=1}^{2T} \geq U_2(\textbf{x}_t,\textbf{y}_t,\textbf{z}_t)_{t=1}^{2T} \text{ and } \\
& U_1(\textbf{x}_t,\textbf{y}_t,\textbf{z}_t)_{t=1}^{2T} \geq U_3(\textbf{x}_t,\textbf{y}_t,\textbf{z}_t)_{t=1}^{2T}
\end{align*}
since the manipulator uses a winning batch coordination policy. 
In conclusion, 
\begin{align*}
&\mathbb{P}\bigg(\Big\{\omega \in \Omega : \exists T_0 = T_0(\omega) \in \mathbb{N} \; \forall T>T_0,\\
& U_1(\textbf{x}_t,\textbf{y}_t,\textbf{z}_t)_{t=1}^{T}(\omega) \geq U_2(\textbf{x}_t,\textbf{y}_t,\textbf{z}_t)_{t=1}^{T}(\omega) \text{ and } \\
& U_1(\textbf{x}_t,\textbf{y}_t,\textbf{z}_t)_{t=1}^{T}(\omega) \geq U_3(\textbf{x}_t,\textbf{y}_t,\textbf{z}_t)_{t=1}^{T}(\omega)\Big \} \bigg) = 1
\end{align*}

where $\Omega$ is the set of all action sequences that could possibly be played out by the three players from time $t=1$ to $\infty$. 

\end{proof}

\begin{proof}[Proof of Theorem \ref{thm5}]

For $(\textbf{e}_{i_1^*},\textbf{e}_{j_1^*},\textbf{e}_{k_1^*})$ and  $(\textbf{e}_{i_2^*},\textbf{e}_{j_2^*},\textbf{e}_{k_2^*})$ with $i_1^*,i_2^* \in [n], j_1^*,j_2^* \in [m], k_1^*,k_2^* \in [l]$, we obtain the winning strategy as a feasible solution of one of the $n^2 \times m^2 \times l^2$ linear feasibility problems with the set of variables\\ $V = \{\hat{d}_2,\hat{d}_3,\tilde{d}_2,\tilde{d}_3,\hat{v}_{2,k},\hat{v}_{3,j},\tilde{v}_{2,k},\tilde{v}_{3,j} \hat{A}^{(2,1)}_{i,j}, \hat{A}^{(3,1)}_{i,k},\tilde{A}^{(2,1)}_{i,j}, \tilde{A}^{(3,1)}_{i,k} \; :\; i \in [n], \; j \in [m], \; k \in [l]\}$:\\
The linear constraints of each linear feasibility problem are:
\begin{align*}
& {|\hat{A}^{(2,1)}(i,j) - A_0^{(2,1)}(i,j)|}{\leq \hat{d}_2 \;\; \forall (i,j) \in [n]\times[m]}{}\\
& {|\hat{A}^{(3,1)}(i,k) - A_0^{(3,1)}(i,k)|}{\leq \hat{d}_3 \;\; \forall (i,k) \in [n]\times[l]}{}\\
& {\hat{A}^{(2,1)}(i_1^*,j) + A^{(2,3)}_0(j,k)}{=  \hat{v}_{2,k} \text{ for} \; \; j=j_1^*\; \; k \in [l]}{}\\
& {\hat{A}^{(2,1)}(i_1^*,j) + A^{(2,3)}_0(j,k)}{\leq  \hat{v}_{2,k} - \epsilon \text{ for} \; \; j\neq j_1^* \; \; k \in [l]}{}\\
& {\hat{A}^{(3,1)}(i_1^*,k) + A^{(3,2)}_0(j,k)}{=  \hat{v}_{3,j} \text{ for} \; \; k= k_1^* \; \; j \in [m]}{}\\
& {\hat{A}^{(3,1)}(i_1^*,k) + A^{(3,2)}_0(j,k)}{\leq  \hat{v}_{3,j} - \epsilon \text{ for} \; \; k\neq k_1^* \; \; j \in [m]}{}\\
\end{align*}
\begin{align*}
& {|\tilde{A}^{(2,1)}(i,j) - A_0^{(2,1)}(i,j)|}{\leq \tilde{d}_2 \;\; \forall (i,j) \in [n]\times[m]}{}\\
& {|\tilde{A}^{(3,1)}(i,k) - A_0^{(3,1)}(i,k)|}{\leq \tilde{d}_3 \;\; \forall (i,k) \in [n]\times[l]}{}\\
& {\tilde{A}^{(2,1)}(i_2^*,j) + A^{(2,3)}_0(j,k)}{=  \tilde{v}_{2,k} \text{ for} \; \; j=j_2^* \; \; k \in [l]}{}\\
& {\tilde{A}^{(2,1)}(i_2^*,j) + A^{(2,3)}_0(j,k)}{\leq  \tilde{v}_{2,k} - \epsilon \text{ for} \; \; j\neq j_2^* \; \; k \in [l]}{}\\
& {\tilde{A}^{(3,1)}(i_2^*,k) + A^{(3,2)}_0(j,k)}{=  \tilde{v}_{3,j} \text{ for} \; \; k= k_2^* \; \; j \in [m]}{}\\
& {\tilde{A}^{(3,1)}(i_2^*,k) + A^{(3,2)}_0(j,k)}{\leq  \tilde{v}_{3,j} - \epsilon \text{ for} \; \; k\neq k_2^* \; \; j \in [m]}{}\\    
\end{align*}
\begin{align*}
&\hat{v}_{2,k_1^*}+\tilde{v}_{2,k_2^*} +\epsilon \leq \hat{A}^{(1,2)}_{i_1^*,j_1^*} + \hat{A}^{(1,3)}_{i_1^*,k_1^*} - \hat{d}_2 - \hat{d}_3\\
& + \tilde{A}^{(1,2)}_{i_2^*,j_2^*} + \tilde{A}^{(1,3)}_{i_2^*,k_2^*} - \tilde{d}_2 - \tilde{d}_3\\
& \hat{v}_{3,j_1^*} + \tilde{v}_{3,j_2^*} +\epsilon\leq \hat{A}^{(1,2)}_{i_1^*,j_1^*} + \hat{A}^{(1,3)}_{i_1^*,k_1^*} - \hat{d}_2 - \hat{d}_3 \\
& + \tilde{A}^{(1,2)}_{i_2^*,j_2^*} + \tilde{A}^{(1,3)}_{i_2^*,k_2^*} - \tilde{d}_2 - \tilde{d}_3      
\end{align*}

The first and second six set of constraints represents the linear constraints required to find the game matrices for the first and second respective halves of the game. Note that in the last two constraints $\hat{A}^{(1,2)}_{i_1^*,j_1^*} + \hat{A}^{(1,3)}_{i_1^*,k_1^*} - \hat{d}_2 - \hat{d}_3$ is the single shot utility for player 1 in the first half, $\hat{v}_{2,k_1^*}$ is the single-shot utility for player 2 in the first half and  $\hat{v}_{2,k_1^*}$ is the single-shot utility for player 3 in the first half under the strategy profile $(\textbf{e}_{i_1},\textbf{e}_{j_1},\textbf{e}_{k_1})$. Similarly the single-shot utilities for the three players in the second half under the strategy profile $(\textbf{e}_{i_2},\textbf{e}_{j_2},\textbf{e}_{k_2})$ can be identified. 
Therefore, the last two constraints represents the constraint required for the manipulator to win, from the definition of batch coordination policy which is
\begin{align*}
&\widehat{U}_1(\textbf{e}_{i_2},\textbf{e}_{j_2},\textbf{e}_{k_2})+ \widehat{U}_1(\textbf{e}_{i_3},\textbf{e}_{j_3},\textbf{e}_{k_3})  >\widehat{U}_i(\textbf{e}_{i_3},\textbf{e}_{j_3},\textbf{e}_{k_3}) + \widehat{U}_i(\textbf{e}_{i_3},\textbf{e}_{j_3},\textbf{e}_{k_3})
\end{align*}
\end{proof}

\begin{lemma}
\label{lem3}
If a player is persistent, then she is consistent.
\end{lemma}
\begin{proof}
Assume player 3 is persistent. Suppose that there exists an action $k^*$ for player 3 that is the unique best response for her for every round of the game. Suppose that within $T$ rounds of the game the number of rounds she plays action $k^*$ is $T_3$. Note that \[\mathbb{P}\Big(\textbf{e}_{k^*} = \argmax_{\textbf{z} \in \Delta_l}{U_3(\textbf{x}_{t},\textbf{y}_{t},\textbf{z})_{t=1}^{T}} \text{ eventually}\Big)=1\] is satisfied trivially since $k^*$ is the unique best response for every round, so it the unique best response in hindsight from round 1 onwards. Since player 3 is persistent, this implies \[\mathbb{P}\Big(\lim_{T\rightarrow\infty}{\frac{T^*}{T}} = 1\Big)=1\]
\end{proof}

\begin{lemma}
\label{lem4}
If a player is no-regret, then she is persistent.
\end{lemma}
\begin{proof}
Suppose player 3 is no-regret. Assume that \[\mathbb{P}\Big(\textbf{e}_{k^*} = \argmax_{\textbf{z} \in \Delta_l}{U_3(\textbf{x}_{t},\textbf{y}_{t},\textbf{z})_{t=1}^{T}} \text{ eventually}\Big)=1\] This implies \[\mathbb{P}\Big(\textbf{e}_{k^*} = \argmax_{\textbf{z} \in \Delta_l}{U_3(\textbf{x}_{t},\textbf{y}_{t},\textbf{z})_{t=1}^{\infty}} \Big)=1\]
By Markov's inequality for any $\epsilon > 0$.
\begin{align*}
& \mathbb{P}\left(U_3(\textbf{x}_t,\textbf{y}_{t},\textbf{e}_{k^*})_{t=1}^{\infty} - U_3(\textbf{x}_t,\textbf{y}_t,\textbf{z}_t)_{t=1}^{\infty} \leq \epsilon\right) \\
& = \mathbb{P}\left(\max_{z \in \Delta_l}[U_3(\textbf{x}_t,\textbf{y}_{t},\textbf{z})_{t=1}^{\infty}] - U_3(\textbf{x}_t,\textbf{y}_t,\textbf{z}_t)_{t=1}^{\infty} \leq \epsilon\right) \\
& > 1 - \frac{\mathbb{E}\left(\max_{z \in \Delta_l}[U_3(\textbf{x}_t,\textbf{y}_{t},\textbf{z})_{t=1}^{\infty}] - U_3(\textbf{x}_t,\textbf{y}_t,\textbf{z}_t)_{t=1}^{\infty}\right)}{\epsilon} \\
& = 1 - \frac{0}{\epsilon} = 1
\end{align*}
where the first equality follows from the definition of a No-Regret Algorithm.
Define \[\mathcal{B}_n := \left\{U_3(\textbf{x}_t,\textbf{y}_{t},\textbf{e}_{k^*})_{t=1}^{\infty} - U_3(\textbf{x}_t,\textbf{y}_t,\textbf{z}_t)_{t=1}^{\infty} \leq 1/n \right\}\] We know that $\mathbb{P}(\mathcal{B}_n) = 1$ for all $n \in \mathbb{N}$. Further, 
\begin{align*}
& \mathbb{P}\left(U_3(\textbf{x}_t,\textbf{y}_t,\textbf{z}_t)_{t=1}^{\infty} = U_3(\textbf{x}_t,\textbf{y}_{t},\textbf{e}_{k^*})_{t=1}^{\infty}\right) = \mathbb{P}\left(\lim_{n \rightarrow \infty}{\mathcal{B}_n}\right) \\
&  = \mathbb{P}\left(\bigcap_{n \in \mathbb{N}}\mathcal{B}_n\right) = 1 - \mathbb{P}\left(\bigcap_{n \in \mathbb{N}}{\mathcal{B}_n}^c\right) = 1 - \mathbb{P}\left(\bigcup_{n \in \mathbb{N}}{\mathcal{B}_n^c}\right)\\
&  \geq 1 - \sum_{n \in \mathbb{N}}{\mathbb{P}\left({\mathcal{B}_n^c}\right)}  = 1 - \sum_{n \in \mathbb{N}}{1 - \mathbb{P}\left({\mathcal{B}_n}\right)} = 1\\
\end{align*}

Assume $U_3(\textbf{x}_t,\textbf{y}_t,\textbf{z}_t)_{t=1}^{\infty} = U_3(\textbf{x}_t,\textbf{y}_{t},\textbf{e}_{k^*})_{t=1}^{\infty}$. Let $T_3$ be the number of rounds player 3 plays action $k^*$ within $T$ rounds. Then $\lim_{T \rightarrow \infty}{\frac{T_3}{T}} = 1$, otherwise we arrive at a contradiction that the long-run utility of P3 is less than $U_3(\textbf{x}_t,\textbf{y}_{t},\textbf{e}_{k^*})_{t=1}^{\infty}$. This is because $k^*$ is the unique best action for P3.

In conclusion, 
\begin{align*}
& \mathbb{P}\Big(\lim_{T \rightarrow \infty}{\frac{T_3}{T}} = 1 \Big) \\
& \geq \mathbb{P}\Big(\textbf{e}_{k^*} = \argmax_{\textbf{z} \in \Delta_l}{U_3(\textbf{x}_t,\textbf{y}_{t},\textbf{z})_{t=1}^{\infty}} \Big)=1
\end{align*}

\end{proof}

\begin{proof}[Proof of Proposition \ref{prop2}]
The proof is an immediate consequence of Lemma \ref{lem3} and Lemma \ref{lem4}
\end{proof}

\begin{proof}[Proof of Proposition \ref{prop3}]
It is a well known result that FTL is not a no-regret algorithm, this is proven in ~\cite{Nicolo06}. Now suppose player 3 uses FTL to play the game and that \[\mathbb{P}\Big(\textbf{e}_{k^*} = \argmax_{\textbf{z} \in \Delta_l}{U_3(\textbf{x}_{t},\textbf{y}_{t},\textbf{z})_{t=1}^{T}} \text{ eventually}\Big)=1\] That is, \[\mathbb{P}\Big(\exists T_0 \in \mathbb{N}: \forall T>T_0 \; \textbf{e}_{k^*} = \argmax_{\textbf{z} \in \Delta_l}{U_3(\textbf{x}_{t},\textbf{y}_{t},\textbf{z})_{t=1}^{T}}\Big)=1\] Therefore there exists a finite-time cutoff point after which FTL will only play $k^*$ for the rest of time (almost surely). This implies \[\mathbb{P}\Big(\lim_{T\rightarrow\infty}{\frac{T_3}{T}} = 1\Big)=1\] where $T_3$ is the number of rounds within $T$ rounds that player 3 played action $k^*$.
\end{proof}

\begin{proof}[Proof of Theorem \ref{thm6}]
Now we find the strategy profile and matrices associated with the largest margin dominance solvable policy.
For $(\textbf{e}_{i^*},\textbf{e}_{j^*},\textbf{e}_{k^*})$ with $i^* \in [n], j^* \in [m], k^* \in [l]$, we obtain them from the optimal solution of one of the $n \times m \times l$ linear programs with the set of variables\\ $V = \{d_2,d_3,v_0,v_{2,k},v_{3,j}, A^{(2,1)}_{i,j}, A^{(3,1)}_{i,k} \; :\; i \in [n], \; j \in [m], \; k \in [l]\}$:
\begin{maxi*}|s|<b>
{\scriptstyle V}{v_0}{}{}
\addConstraint {|A^{(2,1)}(i,j) - A_0^{(2,1)}(i,j)|}{\leq d_2 \;\; \forall (i,j) \in [n]\times[m]}{}
\addConstraint {|A^{(3,1)}(i,k) - A_0^{(3,1)}(i,k)|}{\leq d_3 \;\; \forall (i,k) \in [n]\times[l]}{}
\addConstraint {A^{(2,1)}(i^*,j) + A^{(2,3)}_0(j,k)}{=  v_{2,k} \text{ for} \; \; j=j^* \; \; k \in [l]}{}
\addConstraint {A^{(2,1)}(i^*,j) + A^{(2,3)}_0(j,k)}{\leq  v_{2,k} - \epsilon \text{ for} \; \; j\neq j^* \; \; k \in [l]}{}
\addConstraint {A^{(3,1)}(i^*,k) + A^{(3,2)}_0(j,k)}{=  v_{3,j} \text{ for} \; \; k= k^* \; \; j \in [m]}{}
\addConstraint {A^{(3,1)}(i^*,k) + A^{(3,2)}_0(j,k)}{\leq  v_{3,j} - \epsilon \text{ for} \; \; k\neq k^* \; \; j \in [m]}{}
\addConstraint {v_0}{ \leq A^{(1,2)}_{i^*,j^*} + A^{(1,3)}_{i^*,k^*} - d_2 - d_3 - v_{2,k^*}}{ }
\addConstraint {v_0}{ \leq A^{(1,2)}_{i^*,j^*} + A^{(1,3)}_{i^*,k^*} - d_2 - d_3 - v_{3,j^*}}{ }
\addConstraint {v_0}{ \geq 0}{ }
\end{maxi*}
The only difference in analysis from the linear constraints from Lemma \ref{lem5} and this one is that the last three set of constraints correspond maximizing the margin whilst winning instead of just winning. In the final set of constraints, $v_0$ represents a lower bound on the margin, and therefore the maximization objective is  $v_0$. We obtain the strategy and matrices from the optimal solution that has the greatest value for $v_0$ out of all the $n \times m \times l$ LPs.    
\end{proof}
If we remove the constraint $v_0 \geq 0$ from the above proof, then we obtain a winning dominance solvable policy with the largest margin if a winning dominance solvable policy exists. We obtain a large margin losing strategy otherwise. With the constraint $v_0 \geq 0$, the requirement of deriving a winning strategy is made explicit.

\begin{proof}[Proof of Theorem \ref{thm7}]
We can find the dominance solvable policy that wins with the lowest inefficiency ratio by running $n\times m\times l$ LPs, where the linear constraints are the ones used for winning dominance solvable type-\RN{1} policies, and the minimization objective is $f(V) = d_2 + d_3$. We then finding the optimal solution that minimizes \[\frac{f(V)}{A^{(1,2)}_{i^*,j^*} + A^{(1,3)}_{i^*,k^*}}\] and infer the associated strategy profile and matrices from the solution.
\end{proof}

\begin{proof}[Proof of Theorem \ref{thm8}]
Now we find the strategy profile and matrices associated with the dominance solvable policy that maximizes the egalitarian social welfare.
For $(\textbf{e}_{i^*},\textbf{e}_{j^*},\textbf{e}_{k^*})$ with $i^* \in [n], j^* \in [m], k^* \in [l]$, we obtain them from the optimal solution of one of the $n \times m \times l$ linear programs with the set of variables\\ $V = \{d_2,d_3,v_0,v_{2,k},v_{3,j}, A^{(2,1)}_{i,j}, A^{(3,1)}_{i,k} \; :\; i \in [n], \; j \in [m], \; k \in [l]\}$:
\begin{maxi*}|s|<b>
{\scriptstyle V}{v_0}{}{}
\addConstraint {|A^{(2,1)}(i,j) - A_0^{(2,1)}(i,j)|}{\leq d_2 \;\; \forall (i,j) \in [n]\times[m]}{}
\addConstraint {|A^{(3,1)}(i,k) - A_0^{(3,1)}(i,k)|}{\leq d_3 \;\; \forall (i,k) \in [n]\times[l]}{}
\addConstraint {A^{(2,1)}(i^*,j) + A^{(2,3)}_0(j,k)}{=  v_{2,k} \text{ for} \; \; j=j^* \; \; k \in [l]}{}
\addConstraint {A^{(2,1)}(i^*,j) + A^{(2,3)}_0(j,k)}{\leq  v_{2,k} - \epsilon \text{ for} \; \; j\neq j^* \; \; k \in [l]}{}
\addConstraint {A^{(3,1)}(i^*,k) + A^{(3,2)}_0(j,k)}{=  v_{3,j} \text{ for} \; \; k= k^* \; \; j \in [m]}{}
\addConstraint {A^{(3,1)}(i^*,k) + A^{(3,2)}_0(j,k)}{\leq  v_{3,j} - \epsilon \text{ for} \; \; k\neq k^* \; \; j \in [m]}{}
\addConstraint {v_0}{\leq A^{(1,2)}_{i^*,j^*} + A^{(1,3)}_{i^*,k^*} - d_2 - d_3}{ }
\addConstraint {v_0}{\leq v_{2,k^*}}
\addConstraint {v_0}{\leq v_{3,j^*}}{ }
\end{maxi*}

The analysis of the first six set of linear constraints are the same as that in Lemma \ref{lem5}. In the final three constraints we ensure that the utility of all three players are at least $v_0$. Hence $v_0$ is at most the egalitarian social welfare. We then maximize $v_0$ which maximizes the egalitarian social welfare since it is a lower bound on it.
\end{proof}

\section{Further related work}
\label{app:further related work}

As mentioned earlier, the problem of constructing zero-sum games with a pre-specified (strictly) dominant strategy is similar to designing games with unique minimax equilibrium. The latter was first considered by~\cite{bohnenblust1950}, where an algorithm for constructing a zero-sum game with a unique minimax equilibrium was provided. 

This problem of certifying the uniqueness of linear programming solutions has been studied within the optimisation community. Appa~\cite{appa2002uniqueness} provides a constructive method for verifying the uniqueness of an LP solution which requires solving an addition linear program. Mangasarian~\cite{mangasarian1979uniqueness} describes a number of conditions which guarantee the uniqueness of an LP solution. In fact, it is one of these conditions that we shall leverage to derive our methods for constructing zero-sum games with unique solutions. More generally, many other works deal with characterising the optimal solution sets of linear programs \cite{Kantor, Tavas}, but do not typically pay any special attention to uniqueness.

Note that, since the work of ~\cite{bohnenblust1950}, the uniqueness of Nash equilibrium points has also been studied extensively within the game theory literature~\cite{jansen, karlin, millham, Kreps1981, heuer75, jansenregularity}. For zero-sum games, a necessary and sufficient dimensionality relationship between the unique equilibrium strategy of each player was described in the seminal work of \cite{bohnenblust1950}. Subsequently, similar conditions have been derived for bimatrix games. \cite{rag} derived necessary and sufficient conditions for the existence of bimatrix games with a given unique completely mixed equilibrium. This work was extended by \cite{Kreps1974} to encompass all mixed strategies. More generally, for differentiable n-player concave games, \cite{rosen65} provided a sufficient condition for uniqueness known as diagonal strict concavity.



\end{document}